\documentclass[11pt,letterpaper,reqno]{amsart}

\title[Relative Dynamics of Confined {BEC} Vortices]{Relative Dynamics of Vortices\\in Confined Bose--Einstein Condensates}
\author{Tomoki Ohsawa}
\address{Department of Mathematical Sciences, The University of Texas at Dallas, 800 W Campbell Rd, Richardson, TX 75080-3021}
\email{tomoki@utdallas.edu}
\date{\today}


\usepackage{graphicx,amssymb,paralist}

\usepackage{caption,subcaption}
\usepackage[margin=1in, marginpar=.5in]{geometry}

\usepackage{tikz}
\usetikzlibrary{tikzmark,fit}
\usepackage{eso-pic}
\usepackage{tikz-cd}

\usepackage{nicematrix}

\usepackage[numbers,sort&compress]{natbib}
\usepackage[colorlinks=false]{hyperref}

\usepackage[nameinlink]{cleveref}

\theoremstyle{plain}
\newtheorem{theorem}{Theorem}[section]

\newtheorem{proposition}[theorem]{Proposition}

\theoremstyle{definition}

\theoremstyle{remark}
\newtheorem{remark}[theorem]{Remark}

\def\od#1#2{\dfrac{d#1}{d#2}}
\def\pd#1#2{\dfrac{\partial #1}{\partial #2}}
\def\fd#1#2{\dfrac{\delta #1}{\delta #2}}

\def\tfd#1#2{\delta #1/\delta #2}

\def\parentheses#1{\!\left(#1\right)}
\def\brackets#1{\!\left[#1\right]}
\def\braces#1{\!\left\{#1\right\}}

\def\tr{\mathop{\mathrm{tr}}\nolimits}
\def\Span{\operatorname{span}} % span << WARNING: Defining "span" makes gather and align impossible to be used. Hence "Span."
\def\rank{\operatorname{rank}}
\def\diag{\operatorname{diag}}

\def\norm#1{\left\|#1\right\|}

\def\R{\mathbb{R}}
\def\C{\mathbb{C}}

\def\defeq{\mathrel{\mathop:}=}

\def\setdef#1#2{ \left\{ #1 \ |\ #2 \right\} }
\def\ip#1#2{\left\langle#1,#2\right\rangle}

\renewcommand{\Re}{\operatorname{Re}}
\renewcommand{\Im}{\operatorname{Im}}

\def\rmi{{\rm i}}

\def\d{\mathbf{d}}

\def\PB#1#2{\left\{#1,#2\right\}}

\newcommand\ad{\operatorname{ad}}

\def\U{\mathsf{U}}

\def\u{\mathfrak{u}}

\def\bx{\mathbf{x}}

\begin{document}

\footskip=.6in

\begin{abstract}
  We consider the relative dynamics---the dynamics modulo rotational symmetry in this particular context---of $N$ vortices in confined Bose--Einstein Condensates (BEC) using a finite-dimensional vortex approximation to the two-dimensional Gross--Pitaevskii equation. We give a Hamiltonian formulation of the relative dynamics by showing that it is an instance of the Lie--Poisson equation on the dual of a certain Lie algebra. Just as in our accompanying work on vortex dynamics with the Euclidean symmetry, the relative dynamics possesses a Casimir invariant and evolves in an invariant set, yielding an Energy--Casimir-type stability condition.
  We consider three examples of relative equilibria---those solutions that are undergoing rigid rotations about the origin---with $N=2, 3, 4$, and investigate their stability using the stability condition.
\end{abstract}

\maketitle

\section{Introduction}
\subsection{Dynamics of Confined BEC Vortices}
We consider the dynamics of $N$ interacting vortices $\{ \bx_{i} = (x_{i},y_{i}) \in \R^{2} \}_{i=1}^{N}$ with topological charges $\{ \Gamma_{i} \in \mathbb{Z} \backslash\{0\} \}_{i=1}^{N}$ in a harmonic trap on the plane $\R^{2}$ governed by
\begin{equation}
  \label{eq:basic}
  \begin{split}
    \dot{x}_{i} &= -\Gamma_{i} \frac{y_{i}}{1 - \norm{\bx_{i}}^{2}} - c \sum_{\substack{1\le j \le N\\ j \neq i}} \Gamma_{j} \frac{y_{i} - y_{j}}{ \norm{\bx_{i} - \bx_{j}}^{2} }, \\
    \dot{y}_{i} &=  \Gamma_{i} \frac{x_{i}}{1 - \norm{\bx_{i}}^{2}} + c \sum_{\substack{1\le j \le N\\ j \neq i}} \Gamma_{j} \frac{x_{i} - x_{j}}{ \norm{\bx_{i} - \bx_{j}}^{2} }
  \end{split}
\end{equation}
with some constant $c > 0$ (see below for its definition) for $i \in \{1, \dots, N\}$.
In what follows we shall often identify $\R^{2}$ with $\C$ in the standard manner.
Indeed, one may write the above equations in a more succinct form via $z_{i} \defeq x_{i} + \rmi y_{i} \in \C$ as follows:
\begin{equation}
  \label{eq:N-point_vortices}
  \dot{z}_{i} = \rmi\parentheses{
    \Gamma_{i}\frac{z_{i}}{1 - |z_{i}|^{2}}
    + c \sum_{\substack{1\le j \le N\\ j \neq i}} \Gamma_{j} \frac{z_{i} - z_{j}}{|z_{i} - z_{j}|^{2}}
  }
  \qquad
  i \in \{1, \dots, N\}.
\end{equation}

These equations are obtained as a finite-dimensional vortex approximation to the Gross--Pitaevskii (GP) equation for quasi-two-dimensional (pancake-shaped) Bose--Einstein condensates (BEC) confined by a harmonic potential (see, e.g., \citet{FeSv2001,Fe2009} and references therein):
\begin{equation}
  \label{eq:GP}
  \rmi \pd{\psi}{\tau} = -\frac{1}{2} \parentheses{ \pd{^{2}}{\xi^{2}} +  \pd{^{2}}{\eta^{2}} } \psi + \frac{\omega_{\text{tr}}^{2}}{2}(\xi^{2} + \eta^{2}) + (|\psi|^{2} - \mu) \psi
\end{equation}
for $\psi\colon \R \times \R^{2} \to \C; (\tau, (\xi,\eta)) \mapsto \psi(\tau, \xi, \eta)$, where $\mu$ is the chemical potential and $\omega_{\text{tr}} > 0$ is the ratio of the transversal (i.e., along the $(\xi,\eta)$-plane) frequency of the three-dimensional harmonic trapping potential to its longitudinal (i.e., perpendicular to the $(\xi,\eta)$-plane) frequency.

More specifically, \citet{MiKeFrCaSc2010a,MiKeFrCaSc2010b,MiToKeFrCaScFrHa2011} (see also \cite{KeCaFrKe2004}) applied to \eqref{eq:GP} those techniques developed for vortex dynamics in (untrapped) nonlinear Schr\"odinger equation (see, e.g., \citet{Ne1990} and \citet{PiRu1991}) and derived the equations of motion for the centers $\{ \boldsymbol{\xi}_{i} \defeq (\xi_{i},\eta_{i}) \in \R^{2} \}_{i=1}^{N}$ of $N$ vortices in the confined BEC as
\begin{align*}
  \od{\xi_{i}}{\tau} &= -\Gamma_{i} \omega_{\text{pr}}^{0} \frac{\eta_{i}}{1 - \norm{\boldsymbol{\xi}_{i}}^{2}} - b \sum_{\substack{1\le j \le N\\ j \neq i}} \Gamma_{j} \frac{\eta_{i} - \eta_{j}}{ \norm{\boldsymbol{\xi}_{i} - \boldsymbol{\xi}_{j}}^{2} }, \\
  \od{\eta_{i}}{\tau} &=  \Gamma_{i} \omega_{\text{pr}}^{0} \frac{\xi_{i}}{1 - \norm{\boldsymbol{\xi}_{i}}^{2}} + b \sum_{\substack{1\le j \le N\\ j \neq i}} \Gamma_{j} \frac{\xi_{i} - \xi_{j}}{ \norm{\boldsymbol{\xi}_{i} - \boldsymbol{\xi}_{j}}^{2} }
\end{align*}
with
\begin{equation*}
  \omega_{\text{pr}}^{0} \defeq \frac{1}{R_{\text{TF}}^{2}} \ln\parentheses{ 2\sqrt{2}\pi \frac{\mu}{\omega_{\text{tr}}} },
  \qquad
  R_{\text{TF}} \defeq \frac{\sqrt{2\mu}}{\omega_{\text{tr}}},
\end{equation*}
which are the precession frequency at the trap center (the origin of the $(\xi,\eta)$-plane), and the Thomas--Fermi radius, i.e., an approximate radial extent of the pancake-shaped BEC; the numerical factor $b$ characterizes the strength of interactions between vortices; for example, in the experimental setting of \cite{MiToKeFrCaScFrHa2011}, it is determined empirically that $b = 1.35$.

These equations are inspired by experimental observations of vortex dipoles~\cite{NeSaBrDaAn2010,FrBiKaLaHa2010} and three-vortex configurations~\cite{SeHeHaShRaCaCaCaTaPoRoMaBa2010}, have shown a good agreement with experiments for vortex dipoles~\cite{MiToKeFrCaScFrHa2011,NaCaToKeFrRaAlHa2013}, and are also mathematically justified by \citet{PeKe2011} for $N = 1,2,4$ as a variational approximation to the GP equation.

The equations in \eqref{eq:basic} are obtained from the above equations via rescaling (see, e.g., \citet{KoVoKe2014})
\begin{equation*}
  (x_{i}, y_{i}) \defeq \frac{1}{R_{\text{TF}}} (\xi_{i}, \eta_{i}),
  \qquad
  t \defeq \omega_{\text{pr}}^{0} \tau,
  \qquad
  c \defeq \frac{b}{ 2\ln\parentheses{ 2\sqrt{2}\pi \mu/\omega_{\text{tr}} } },
\end{equation*}
and we shall focus on the dynamics of \eqref{eq:basic} in this paper.

We note that we are interested in the dynamics of vortices trapped in the open unit disc centered at the origin of $\R^{2} \cong \C$, i.e., $\norm{\bx_{i}} = |z_{i}| < 1$ for every $i \in \{1, \dots, N\}$---within the Thomas--Fermi radius from the origin.

The dynamics of \eqref{eq:basic} has been studied fairly well for a few vortices ($N = 2,3,4$):
For $N = 2$, the dipole case with $\Gamma_{1} = 1$ and $\Gamma_{2} = -1$ was studied theoretically and numerically by \citet{GoKeCa2015} and \citet{ToKeFrCaScHa2011}, and the same sign case $\Gamma_{1} = \Gamma_{2} = 1$ numerically by \citet{MuGrKuSi2016}.
For $N = 3$, the chaotic dynamics of the tripole case with $(\Gamma_{1}, \Gamma_{2}, \Gamma_{3}) = (1, -1, 1)$ was studied numerically by \citet{KyKoSkKe2014}, and its transition to chaos was studied by \citet{KoVoKe2014} combining analytical and numerical methods.
Also, \citet{NaCaToKeFrRaAlHa2013} revealed a pitchfork bifurcation behind the instability of the symmetric configurations of a few vortices experimentally as well as using a combined theoretical and numerical method.

\subsection{Hamiltonian Formulation}
It is well known that the equations in \eqref{eq:basic} constitute a Hamiltonian system in the sense we shall describe below.

Let
\begin{equation}
  \label{eq:D_Gamma}
  \mathsf{D}_{\Gamma} \defeq \diag(\Gamma_{1}, \dots, \Gamma_{N})
\end{equation}
be the $N \times N$ diagonal matrix whose diagonal entries are the topological charges $\{ \Gamma_{i} \in \mathbb{Z} \backslash\{0\} \}_{i=1}^{N}$, and define the skew-symmetric $2N \times 2N$ matrix
\begin{equation*}
  \mathbb{J} \defeq
  \begin{bmatrix}
    0 & \mathsf{D}_{\Gamma} \\
    -\mathsf{D}_{\Gamma} & 0
  \end{bmatrix}
\end{equation*}
This matrix defines the symplectic form
\begin{equation}
  \label{eq:Omega}
  \Omega \defeq \sum_{i=1}^{N}\Gamma_{i}\, \d{x}_{i} \wedge \d{y}_{i}
\end{equation}
on $\R^{2N}$ in the sense that
\begin{equation*}
  \Omega(v,w) = v^{T} \mathbb{J} w
  \quad
  \forall v, w \in \R^{2N},
\end{equation*}
where we shall use
\begin{equation*}
  z = (x_{1}, \dots, x_{N}, y_{1}, \dots, y_{N}) \in \R^{2N}
  \longleftrightarrow
  z = (z_{1}, \dots, z_{N}) \in \C^{N}
\end{equation*}
interchangeably as coordinates for $\R^{2N} \cong \C^{N}$.
In terms of the complex coordinates in $\C^{N}$, we have
\begin{equation*}
  \Omega = -\frac{1}{2} \sum_{i=1}^{N} \Gamma_{i} \Im(\d{z}_{i} \wedge \d{z}_{i}^{*}) = -\d\Theta
\end{equation*}
with
\begin{equation}
  \label{eq:Theta}
  \Theta \defeq -\frac{1}{2} \sum_{i=1}^{N} \Gamma_{i} \Im(z_{i}^{*} \d{z}_{i}),
\end{equation}

Given a (smooth) function $F\colon \R^{2N} \to \R$, we may define the corresponding Hamiltonian vector field $X_{F}$ on $\R^{2N} \cong \C^{N}$ with respect to the above symplectic form as follows:
\begin{equation*}
  X_{F}(z) \defeq (\mathbb{J}^{T})^{-1} DF(z) = -\mathbb{J}^{-1} DF(z),
\end{equation*}
where $DF(z)$ stands for the gradient of $F$ at $z \in \R^{2N}$ as a column vector in $\R^{2N}$, and
\begin{equation*}
  (\mathbb{J}^{T})^{-1}
  = -\mathbb{J}^{-1}
  = \begin{bmatrix}
    0 & \mathsf{D}_{\Gamma}^{-1} \\
    -\mathsf{D}_{\Gamma}^{-1} & 0
  \end{bmatrix},
\end{equation*}
where $\mathsf{D}_{\Gamma}^{-1} = \diag(1/\Gamma_{1}, \dots, 1/\Gamma_{N})$.

Then the corresponding Poisson bracket is
\begin{equation}
  \label{eq:PB}
  \begin{split}
    \PB{F}{H}(z)
    \defeq \Omega(X_{F}, X_{H})(z)
    &= X_{F}(z)^{T} \mathbb{J} X_{H}(z) \\
    &= DF(z)^{T} (\mathbb{J}^{T})^{-1} DH(z) \\
    &= \sum_{i=1}^{N} \frac{1}{\Gamma_{i}}\parentheses{
      \pd{F}{x_{i}} \pd{H}{y_{i}}
      - \pd{F}{y_{i}} \pd{H}{x_{i}}
      }
  \end{split}
\end{equation}
for every pair of smooth $F, H\colon \R^{2N} \to \R$.
One then sees that $X_{F}(z) = \PB{z}{F}$ in the sense that the equality holds for each pair of corresponding components of $X_{F}(z)$ and $z$.

Let us define a Hamiltonian
\begin{equation}
  \label{eq:H}
  \begin{split}
    H(z) &\defeq \frac{1}{2}\parentheses{
           \sum_{i=1}^{N} \Gamma_{i}^{2} \ln\left(1 - \norm{\bx_{i}}^{2} \right)
           - c \sum_{1\le i < j \le N} \Gamma_{i} \Gamma_{j} \ln\norm{\bx_{i} - \bx_{j}}^{2}
           } \\
         &= \frac{1}{2}\parentheses{
           \sum_{i=1}^{N} \Gamma_{i}^{2} \ln\left(1 - |z_{i}|^{2} \right)
           - c \sum_{1\le i < j \le N} \Gamma_{i} \Gamma_{j} \ln |z_{i} - z_{j}|^{2}
           },
  \end{split}
\end{equation}
which is defined in $\R^{2N}$ except for those points of collisions, i.e., $\bx_{i} = \bx_{j}$ with $i \neq j$.
We shall ignore this issue that $H$ is not defined on the entire $\R^{2N}$ as it is not essential in our treatment of relative dynamics.

Then, the equations in \eqref{eq:basic} are given as the following Hamiltonian system:
\begin{equation*}
  \dot{z} = X_{H}(z) = \PB{z}{H},
\end{equation*}
again in the sense that the equation holds for each component.

\section{Relative Dynamics}
\subsection{Relative Dynamics with $N = 2$}
\label{ssec:illustrative_example}
The theory to be developed in this section applies to $N$ vortices in general, but for illustrative purpose, we shall first consider the case with $N = 2$.

The goal of relative dynamics for $N = 2$ is to describe the dynamics of the triangular shape made by the 3 points consisting of the two vortices and the origin, regardless of its orientation; see \Cref{fig:Dumbbell}.
We note in passing that the ``triangles'' include those degenerate cases where the vortices and the origin are on a single line.
\begin{figure}[htbp]
  \centering
  \includegraphics[width=.25\linewidth]{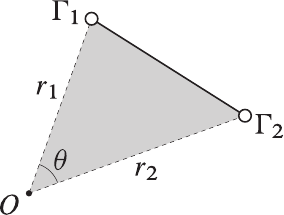}
  \caption{Triangle made by two vortices and the origin, and its parametrization. \citet{ToKeFrCaScHa2011} derived the time evolution equations for $(r_{1}, r_{2}, \theta)$---an instance of relative dynamics for $N = 2$.}
  \label{fig:Dumbbell}
\end{figure}

To put it differently, two triangles obtained by a rigid rotation about the origin are considered the same shape, and we are interested in the time evolution of the triangular shape itself by modding out the rigid rotational motion about the origin.
We shall refer to such dynamics of relative configurations as the \textit{relative dynamics} in what follows.

One sees an instance of relative dynamics for vortex dipoles---$N = 2$ with opposite topological charges $\Gamma_{1} = -\Gamma_{2} = 1$---in \citet{ToKeFrCaScHa2011}:
They rewrite the dynamics of the vortex dipoles located at $z_{i} = r_{i} e^{\rmi\theta_{i}}$ with $i = 1, 2$ into the dynamics of $(r_{1}, r_{2}, \theta)$ with $\theta \defeq \theta_{1} - \theta_{2}$.
Clearly these three parameters describe the shape of the triangle regardless of its orientation; see~\Cref{fig:Dumbbell}.

\citet{ToKeFrCaScHa2011} also wrote down the Hamiltonian $H(r_{1}, r_{2}, \theta)$ in terms of the three parameters.
However, it is not clear how the equations for $(r_{1}, r_{2}, \theta)$ are a Hamiltonian system with this particular Hamiltonian $H$.
Indeed, this is an odd-dimensional system, and so would not be a Hamiltonian system in the canonical sense.

We would like to formulate the relative dynamics as a Hamiltonian system.
To that end, first consider the matrix
\begin{equation}
  \label{eq:mu-N=2}
  \mu =
  \rmi
  \begin{bmatrix}
    \mu_{1} & \mu_{3} + \rmi \mu_{4} \\
    \mu_{3} - \rmi \mu_{4} & \mu_{2}
  \end{bmatrix}
  \defeq \rmi z z^{*} = \rmi
  \begin{bmatrix}
    |z_{1}|^{2} & z_{1} z_{2}^{*} \smallskip\\
    z_{2} z_{1}^{*} &  |z_{2}|^{2}
  \end{bmatrix},
\end{equation}
that is, we have
\begin{equation}
  \label{eq:mu-shapevars_N=2}
  \mu_{1} = r_{1}^{2},
  \qquad
  \mu_{2} = r_{2}^{2},
  \qquad
  \mu_{3} = r_{1} r_{2} \cos\theta,
  \qquad
  \mu_{4} = r_{1} r_{2} \sin\theta.
\end{equation}
We see that the first three $(\mu_{1}, \mu_{2}, \mu_{3})$ are essentially equivalent to $(r_{1}, r_{2}, \theta)$, and parametrize the triangle.
So the last parameter $\mu_{4}$ is redundant; in fact, the above definition of $\mu$ implies that $\rank \mu = 1$ (excluding $\mu = 0$ in which case the vortices collide at the origin) and this gives that
\begin{equation*}
  \det \mu = \mu_{1} \mu_{2} - \mu_{3}^{2} - \mu_{4}^{2} = 0,
\end{equation*}
defining a three-dimensional invariant submanifold of the dynamics, effectively eliminating $\mu_{4}$.
Now, writing the Hamiltonian $H$ from \eqref{eq:H} in terms of $\mu$ as
\begin{equation}
  \label{eq:h-N=2}
  h(\mu) \defeq \frac{1}{2}\parentheses{
    \Gamma_{1}^{2} \ln(1 - \mu_{1}) + \Gamma_{2}^{2} \ln(1 - \mu_{2}) 
    - c\, \Gamma_{1} \Gamma_{2} \ln \left( \mu_{1} + \mu_{2} - 2\mu_{3} \right)
  },
\end{equation}
we can show (as we shall explain below for the general case) that the time evolution of $\mu$ is governed by the matrix differential equation
\begin{equation*}
  \dot{\mu} = \mathsf{D}_{\Gamma}^{-1} \fd{h}{\mu} \mu - \mu \fd{h}{\mu} \mathsf{D}_{\Gamma}^{-1},
\end{equation*}
where
\begin{equation}
  \label{eq:dh-N=2}
  \fd{h}{\mu} \defeq \rmi
  \begin{bmatrix}
    2 \pd{h}{\mu_{1}} & \pd{h}{\mu_{3}} + \rmi\,\pd{h}{\mu_{4}} \medskip\\
      \pd{h}{\mu_{3}} - \rmi\,\pd{h}{\mu_{4}} & 2 \pd{h}{\mu_{2}}
  \end{bmatrix}.
\end{equation}
We shall explain below why the above differential is natural in this context.

The above set of equations is an instance of the special class of Hamiltonian systems called the Lie--Poisson equations; see \citet[Chapter~13]{MaRa1999}.
We shall come back to the case with $N = 2$ in \Cref{ssec:Dumbbell}, and apply the above formulation to the problem of finding relative equilibria and analyzing their stability.

\subsection{Rotational Symmetry}
Consider the $\mathbb{S}^{1}$-action on $\C^{N}$ defined as
\begin{equation}
  \label{eq:S1-action}
  \mathbb{S}^{1} \times \C^{N} \to \C^{N};
  \qquad
  (e^{\rmi\theta}, z) \mapsto e^{\rmi\theta} z = \left( e^{\rmi\theta} z_{1}, \dots, e^{\rmi\theta} z_{N} \right),
\end{equation}
which corresponds to rigid rotations of the $N$ vortices about the origin; see \Cref{fig:Symmetry}.
The system~\eqref{eq:basic} possesses $\mathbb{S}^{1}$-symmetry in the sense that both the symplectic form \eqref{eq:Omega} and the Hamiltonian are invariant under the action.

\begin{figure}[htbp]
  \centering
  \includegraphics[width=.4\linewidth]{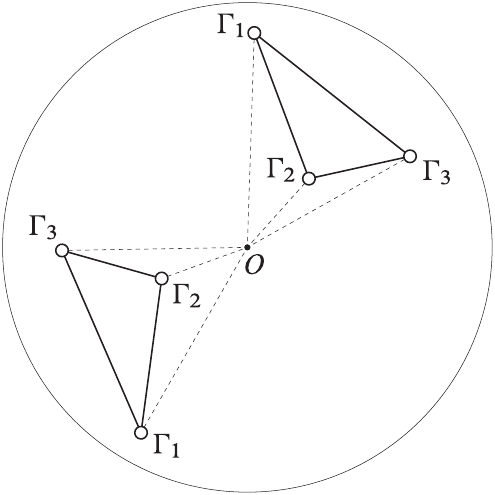}
  \caption{Rotational symmetry pictured for $N = 3$.
    The trapping potential breaks the translational symmetry but retains the rotational symmetry:
    Any two configurations of the 4 points (the vortices and the origin) obtained by a rotation about the origin are essentially the same; they indeed define the same shape.
    In other words, they belong to the same equivalence class defined by the action~\eqref{eq:S1-action}.
  }
  \label{fig:Symmetry}
\end{figure}

This symmetry implies that one may reduce the dynamics by the $\mathbb{S}^{1}$-symmetry to describe the dynamics in which only the relative positions of the $N+1$ points (the $N$ vortices and the origin) matter by identifying those configurations that are rigid rotations to one another as a single relative configuration (i.e., taking an equivalence class).

\subsection{Geometry of Relative Dynamics}
We shall adapt ideas from our previous work~\cite{Oh2019d} (see also \citet{BoPa1998} and \citet{BoBoMa1999}) to our setting, and consider the map
\begin{equation*}
  \C^{N} \to \u(N);
  \qquad
  z \mapsto \rmi z z^{*},
\end{equation*}
sending the positions $z \in \C^{N}$ of the vortices to the skew-Hermitian matrix $\rmi z z^{*}$.
The set of $N \times N$ complex skew-Hermitian matrices is often identified as the Lie algebra
\begin{equation*}
  \u(N) \defeq \setdef{ \xi \in \C^{N\times N} }{ \xi^{*} = -\xi }
\end{equation*}
of the unitary group $\U(N)$ equipped with the Lie bracket given by the standard commutator.
However, we shall instead equip $\u(N)$ with the following non-standard bracket
\begin{equation*}
  [\xi, \eta]_{\Gamma} \defeq \xi \mathsf{D}_{\Gamma}^{-1} \eta - \eta \mathsf{D}_{\Gamma}^{-1} \xi,
\end{equation*}
with $\mathsf{D}_{\Gamma}$ defined in \eqref{eq:D_Gamma}, and define $\u(N)_{\Gamma}$ as the vector space $\u(N)$ equipped with the above bracket; as a result $\u(N)_{\Gamma}$ is a Lie algebra as well.

We shall also define an inner product on $\u(N)_{\Gamma}$ as
\begin{equation}
  \label{eq:ip}
  \ip{\xi}{\eta} \defeq \frac{1}{2}\tr(\xi^{*}\eta),
\end{equation}
and, in what follows, identify the dual $\u(N)_{\Gamma}^{*}$ with $\u(N)_{\Gamma}$ itself via this inner product: An element $\alpha \in \u(N)_{\Gamma}^{*}$ gives a linear map $\alpha\colon \u(N)_{\Gamma} \to \R$, but one can find a unique $\alpha^{\sharp} \in \u(N)_{\Gamma}$ such that $\alpha(\eta) = \ip{\alpha^{\sharp}}{\eta}$ for every $\eta \in \u(N)_{\Gamma}$.

One may also define a Poisson bracket on $\u(N)_{\Gamma}^{*} \cong \u(N)_{\Gamma}$ as follows:
For every pair of smooth $f, h\colon \u(N)^{*} \to \R$,
\begin{equation}
  \label{eq:LPB}
  \begin{split}
    \PB{f}{h}_{\Gamma}(\mu)
    &\defeq \ip{\mu}{ \brackets{ \fd{f}{\mu}, \fd{h}{\mu} }_{\Gamma} } \\
    &= \frac{1}{2}\tr\left( \mu^{*} \left( \fd{f}{\mu} \mathsf{D}_{\Gamma}^{-1} \fd{h}{\mu} - \fd{h}{\mu} \mathsf{D}_{\Gamma}^{-1} \fd{f}{\mu} \right) \right),
  \end{split}
\end{equation}
where the derivative $\tfd{f}{\mu} \in \u(N)_{\Gamma}^{*}$ is defined so that, for any $\mu,\nu \in \u(N)_{\Gamma}^{*}$,
\begin{equation}
  \label{def:tfd}
  \ip{\nu}{\fd{f}{\mu}}
  = \frac{1}{2}\tr\parentheses{ \nu^{*} \fd{f}{\mu} }
  = \left.\od{}{s}\right|_{s=0} f(\mu + s\nu).
\end{equation}
This definition applied to the case with $N = 2$ yields \eqref{eq:dh-N=2}.
The above Poisson bracket~\eqref{eq:LPB} is an instance of the so-called Lie--Poisson bracket defined on the dual of every Lie algebra; see, e.g., \citet[Chapter~13]{MaRa1999}.

Now, we define the map introduced at the beginning of this subsection as 
\begin{equation*}
  \mathbf{J}\colon \C^{N} \to \u(N)_{\Gamma}^{*} \cong \u(N)_{\Gamma};
  \qquad
  z \mapsto \rmi z z^{*}.
\end{equation*}
The significance of this map is that it is a Poisson map with respect to the Poisson bracket~\eqref{eq:PB} on $\R^{2N} \cong \C^{N}$ and the Lie--Poisson bracket~\eqref{eq:LPB} on $\u(N)_{\Gamma}^{*} \cong \u(N)_{\Gamma}$, that is, for every pair of smooth $f, h\colon \u(N)^{*} \to \R$,
\begin{equation*}
  \PB{ f \circ \mathbf{J} }{ h \circ \mathbf{J} }
  = \PB{f}{h}_{\Gamma} \circ \mathbf{J}.
\end{equation*}
One may prove it just as we did in \cite[Section~3.2]{Oh2019d} as follows:
Define the Lie group
\begin{equation*}
  \U(N)_{\Gamma} \defeq \setdef{ U \in \C^{N \times N} }{ U^{*} \mathsf{D}_{\Gamma} U = \mathsf{D}_{\Gamma} },
\end{equation*}
and its action on $\C^{N}$ as follows:
\begin{equation*}
  \U(N)_{\Gamma} \times \C^{N} \to \C^{N};
  \qquad
  (U, z) \mapsto U z.
\end{equation*}
Then this action is symplectic with respect to the symplectic form~\eqref{eq:Omega} as one can easily check using the expression for $\Theta$ in \eqref{eq:Theta}.
Its associated momentum map is then $\mathbf{J}$, and it is equivariant: $U \mathbf{J}(z) U^{*} = \mathbf{J}(U z)$ for every $U \in \U(N)_{\Gamma}$ and every $z \in \C^{N}$.
Then a well-known property of equivariant momentum maps (see, e.g., \cite[Theorem~12.4.1]{MaRa1999}) gives the desired result.

Moreover, writing an arbitrary element $\mu \in \u(N)_{\Gamma}^{*} \cong \u(N)_{\Gamma}$ as
\begin{equation}
  \label{eq:mu}
  \mu = \rmi\,
  \begin{bNiceMatrix}
    \mu_{1} & \mu_{12} & \Cdots & \mu_{1,N} \\
    \mu_{12}^{*} & \Ddots & \Ddots & \vdots \\
    \vdots  & \Ddots & \Ddots  & \mu_{N-1,N} \\
    \mu_{1,N}^{*} & \Cdots & \mu_{N-1,N}^{*} & \mu_{N}
  \end{bNiceMatrix}
\end{equation}
with
\begin{equation*}
  \mu_{i} \in \R \text{ for } 1 \le i \le N,
  \qquad
  \mu_{ij} \in \C \text{ for } 1 \le i < j \le N,
\end{equation*}
we may define
\begin{equation}
  \label{eq:h}
  h(\mu) \defeq \frac{1}{2}\parentheses{
    \sum_{i=1}^{N} \Gamma_{i}^{2} \ln(1 - \mu_{i})
    - c \sum_{1\le i < j \le N} \Gamma_{i} \Gamma_{j} \ln \left( \mu_{i} + \mu_{j} - 2\Re\mu_{ij} \right)
  }
\end{equation}
so that $h \circ \mathbf{J} = H$.

The facts that $\mathbf{J}$ is Poisson and that the original Hamiltonian $H$ may be written as $H = h \circ \mathbf{J}$ implies that $\mathbf{J}$ maps the original Hamiltonian system~\eqref{eq:basic} to another Hamiltonian system.
More specifically, we have
\begin{equation*}
  \dot{z} = \PB{z}{H}
  \overset{\mathbf{J}}{\leadsto}
  \dot{\mu} = \PB{\mu}{h}_{\Gamma}
\end{equation*}
via $\mu = \mathbf{J}(z)$.
We may write down the equations on the right more explicitly as the Lie--Poisson equation
\begin{equation}
  \label{eq:LP}
  \dot{\mu}
  = -\ad^{*}_{\tfd{h}{\mu}} \mu
  = \mathsf{D}_{\Gamma}^{-1} \frac{\delta h}{\delta\mu} \mu - \mu \frac{\delta h}{\delta\mu} \mathsf{D}_{\Gamma}^{-1},
\end{equation}
where we defined the coadjoint action $\ad^{*}\colon \u(N)_{\Gamma} \times \u(N)_{\Gamma}^{*} \to \u(N)_{\Gamma}^{*}$ as follows:
\begin{equation*}
  \ip{ \ad^{*}_{\xi} \mu }{ \eta } = \ip{ \mu }{ [\xi,\eta]_{\Gamma} }
  \quad
  \forall \xi, \eta \in \u(N)_{\Gamma}
  \quad
  \forall \mu \in \u(N)_{\Gamma}^{*},
\end{equation*}
which yields
\begin{equation*}
  \ad^{*}_{\xi} \mu = \mu \xi \mathsf{D}_{\Gamma}^{-1} - \mathsf{D}_{\Gamma}^{-1} \xi \mu,
\end{equation*}
where $\mu \in \u(N)_{\Gamma}^{*}$ is identified with the corresponding element in $\u(N)_{\Gamma}$ via the inner product~\eqref{eq:ip}.

The upshot is that the momentum map $\mathbf{J}$ maps the Hamiltonian system~\eqref{eq:basic} to the Lie--Poisson equation~\eqref{eq:LP}.
This is an instance of the so-called ``collective dynamics''; see, e.g., \citet{GuSt1980} and \cite[Section~28]{GuSt1990}.

\subsection{Casimir of Relative Dynamics}
The Lie--Poisson equation~\eqref{eq:LP} possesses the following Casimir invariant:
\begin{proposition}
  The function
  \begin{equation}
    \label{eq:C}
    C\colon \u(N)_{\Gamma}^{*} \to \R;
    \qquad
    C(\mu) \defeq \tr(-\rmi D_{\Gamma} \mu) = \sum_{i=1}^{N} \Gamma_{i} \mu_{i}
  \end{equation}
  is a Casimir of the Lie--Poisson bracket~\eqref{eq:LPB}, i.e., $\PB{C}{f}_{\Gamma} = 0$ for \textit{every} smooth $f\colon \u(N)_{\Gamma}^{*} \to \R$.
  As a result, $C$ is an invariant of the Lie--Poisson equation~\eqref{eq:LP}.
\end{proposition}
\begin{proof}
  Let us first compute $\tfd{C}{\mu}$ following the definition~\eqref{def:tfd}:
  For any $\mu,\nu \in \u(N)_{\Gamma}^{*}$,
  \begin{equation*}
    \ip{\nu}{\fd{C}{\mu}}
    = \left.\od{}{s}\right|_{s=0} C(\mu + s\nu)
    = \tr(-\rmi \mathsf{D}_{\Gamma} \nu)
    = \ip{ \nu }{ 2\rmi \mathsf{D}_{\Gamma} }
    \implies
    \frac{\delta C}{\delta\mu} = 2\rmi \mathsf{D}_{\Gamma}.
  \end{equation*}
  Therefore,
  \begin{equation*}
    \brackets{ \fd{f}{\mu}, \fd{C}{\mu} }_{\Gamma}
    = 2\rmi \parentheses{
      \fd{f}{\mu} \mathsf{D}_{\Gamma}^{-1} \mathsf{D}_{\Gamma}
      - \mathsf{D}_{\Gamma} \mathsf{D}_{\Gamma}^{-1} \fd{f}{\mu}
    }
    = 0,
  \end{equation*}
  and so
  \begin{equation*}
    \PB{f}{C}_{\Gamma}(\mu)
    = \ip{\mu}{ \brackets{ \fd{f}{\mu}, \fd{C}{\mu} }_{\Gamma} }
    = 0.
  \end{equation*}
  This also implies that $C$ is an invariant of the Lie--Poisson equation~\eqref{eq:LP} because $\dot{C} = \PB{C}{h}_{\Gamma} = 0$.
\end{proof}
\begin{remark}
  The above invariant $C$ is in fact the Noether invariant $\sum_{i=1}^{N} \Gamma_{i} |z_{i}|^{2}$ of the original system~\eqref{eq:basic} associated with the $\mathbb{S}^{1}$-symmetry (called the angular impulse in the point vortex literature) written in terms of $\mu$.
\end{remark}

\subsection{Invariant Set}
One sees in \eqref{eq:mu} that $\mu \in \u(N)_{\Gamma}^{*}$ has $N + 2{N \choose 2} = N^{2}$ real components, and this seems rather redundant to describe the relative configurations made by the origin and the $N$ vortices, particularly when $N$ becomes large.
Since the original dynamics~\eqref{eq:basic} is $2N$-dimensional, the relative dynamics~\eqref{eq:LP} even has a dimension greater than the original one for $N \ge 3$.
So, at the first sight, it does not seem that the relative dynamics is a reduced dynamics of the original system~\eqref{eq:basic}.

It turns out that the apparent increase in the dimension is compensated by the fact that the dynamics~\eqref{eq:LP} evolves in an invariant set of $\u(N)_{\Gamma}^{*}$:

\begin{proposition}
  Consider the subset of $\u(N)_{\Gamma}^{*}$ consisting of those elements that are rank-one:
  \begin{equation*}
    \u_{1}(N)_{\Gamma}^{*} \defeq \setdef{ \mu \in \u(N)_{\Gamma}^{*} \cong \u(N)_{\Gamma} }{ \rank\mu = 1 }.
  \end{equation*}
  Then:
  \begin{enumerate}[(i)]
  \item $\u_{1}(N)_{\Gamma}^{*}$ is an invariant set of the Lie--Poisson dynamics~\eqref{eq:LP}.  \item $\mu \in \u(N)_{\Gamma}^{*} \backslash \{0\}$ satisfies $\rank \mu = 1$ if and only if all the determinants of the $2 \times 2$ submatrices sweeping the upper triangular part and the subdiagonal of $-\rmi\mu$ (see the picture below) vanish.
    \begin{equation*}
      -\rmi\mu = 
      \begin{bmatrix}
        \tikzmarknode{11}{\mu_{11}} & \tikzmarknode{12}{\mu_{12}} & \tikzmarknode{13}{\mu_{13}} & \tikzmarknode{14}{\vphantom{\mu_{14}}\cdots} & \tikzmarknode{15}{\vphantom{\mu_{15}}\cdots} & \tikzmarknode{16}{\mu_{1N}}  \\ 
        \tikzmarknode{21}{\mu_{21}} & \tikzmarknode{22}{\mu_{22}} & \tikzmarknode{23}{\mu_{23}} & \tikzmarknode{24}{\vphantom{\mu_{24}}\cdots} & \tikzmarknode{25}{\vphantom{\mu_{25}}\cdots} & \tikzmarknode{26}{\mu_{2N}} \\
        \tikzmarknode{31}{\mu_{31}} & \tikzmarknode{32}{\mu_{32}} & \tikzmarknode{33}{\mu_{33}} & \tikzmarknode{34}{\vphantom{\mu_{34}}\cdots} & \tikzmarknode{35}{\vphantom{\mu_{35}}\cdots} & \tikzmarknode{36}{\mu_{3N}} \\
        \tikzmarknode{41}{\mu_{41}} & \tikzmarknode{42}{\mu_{42}} & \tikzmarknode{43}{\mu_{43}} & \tikzmarknode{44}{\vphantom{\mu_{44}}\raisebox{0pt}[0pt][0pt]{$\ddots$}} & \tikzmarknode{45}{\vphantom{\mu_{45}}\cdots} & \tikzmarknode{46}{\mu_{4N}} \\
        \vdots & \vdots & \vdots & \vdots & \tikzmarknode{55}{\vphantom{\raisebox{3pt}{$\mu_{55}$}}\raisebox{0pt}[0pt][0pt]{$\ddots$}} & \tikzmarknode{56}{\vphantom{\raisebox{3pt}{$\mu_{56}$}}\raisebox{0pt}[0pt][0pt]{$\vdots$}} \\
        \mu_{N1} & \mu_{N2} & \mu_{N3} & \dots & \cdots & \tikzmarknode{66}{\mu_{NN}}  \\
      \end{bmatrix}
    \end{equation*}
    \AddToShipoutPictureBG*{%
      \begin{tikzpicture}[overlay,remember picture]
        \node[fit=(11) (22),draw=red,fill=red,fill opacity=0.5]{};
        \node[fit=(12) (23),draw=red,fill=red,fill opacity=0.5]{};
        \node[fit=(13) (24),draw=red,fill=red,fill opacity=0.5]{};
        \node[fit=(14) (25),draw=red,fill=red,fill opacity=0.5]{};
        \node[fit=(15) (26),draw=red,fill=red,fill opacity=0.5]{};
        \node[fit=(22) (33),draw=red,fill=red,fill opacity=0.5]{};
        \node[fit=(23) (34),draw=red,fill=red,fill opacity=0.5]{};
        \node[fit=(24) (35),draw=red,fill=red,fill opacity=0.5]{};
        \node[fit=(25) (36),draw=red,fill=red,fill opacity=0.5]{};
        \node[fit=(33) (44),draw=red,fill=red,fill opacity=0.5]{};
        \node[fit=(34) (45),draw=red,fill=red,fill opacity=0.5]{};
        \node[fit=(35) (46),draw=red,fill=red,fill opacity=0.5]{};
        \node[fit=(44) (55),draw=red,fill=red,fill opacity=0.5]{};
        \node[fit=(45) (46) (56),draw=red,fill=red,fill opacity=0.5]{};
        \node[fit=(55) (66),draw=red,fill=red,fill opacity=0.5]{};
      \end{tikzpicture}%
    }
    Specifically, these determinants are given by
    \begin{equation*}
      % \label{eq:Rs}
      \begin{split}
        R_{i}&\colon \u(N)_{\Gamma}^{*} \backslash \{0\} \to \R;
               \quad
               R_{i}(\mu) \defeq
               \begin{vmatrix}
                 \mu_{i} & \mu_{i,i+1} \\
                 \mu_{i,i+1}^{*} & \mu_{i+1}
               \end{vmatrix}
               \quad\text{for}\quad
               1 \le i \le N-1,
        \\
        R_{ij}&\colon \u(N)_{\Gamma}^{*} \backslash \{0\} \to \C;
                \quad
                R_{ij}(\mu) \defeq
                \begin{vmatrix}
                  \mu_{i,j} & \mu_{i,j+1} \\
                  \mu_{i+1,j} & \mu_{i+1,j+1}
                \end{vmatrix}
                \quad\text{for}\quad
                1 \le i < j \le N-1,
      \end{split}
    \end{equation*}
  \item Collect the above determinants to define
    \begin{equation}
      \label{eq:R}
      \begin{split}
        R \colon & \u(N)_{\Gamma}^{*} \backslash \{0\} \to \R^{N-1} \times \C^{N-1 \choose 2} \cong \R^{(N - 1)^{2}}; \\
                 & R(\mu) \defeq (R_{1}(\mu), \dots, R_{N-1}(\mu), R_{12}(\mu), \dots, R_{N-2,N-1}(\mu)).
      \end{split}
    \end{equation}
    Then the invariant set $\u_{1}(N)_{\Gamma}^{*}$ is precisely the level set $R^{-1}(0)$.
  \end{enumerate}
\end{proposition}
\begin{proof}
  \begin{enumerate}[(i)]
  \item The proof goes as in \cite[Proposition~4.1]{Point_Vortex_Stability-Plane}.
    Let $t \mapsto z(t)$ be the solution of the initial value problem of the original evolution equations~\eqref{eq:basic}.
    Then, $\mu(0) \defeq \mathbf{J}(z(0)) = \rmi\,z(0) z(0)^{*}$ gives the corresponding initial condition for the Lie--Poisson equation~\eqref{eq:LP}.
    Now, notice that both $t \mapsto \mu(t)$ and $t \mapsto \mathbf{J}(z(t)) = \rmi\,z(t) z(t)^{*}$ satisfy the same initial value problem for the Lie--Poisson equation~\eqref{eq:LP}.
    Hence by uniqueness we have $\mu(t) = \rmi\,z(t) z(t)^{*}$, and so $\rank \mu(t) = 1$, i.e., $\mu(t) \in \u_{1}(N)_{\Gamma}^{*}$ for every $t$.
  \item This is proved in \citet[Lemma~4.2]{Point_Vortex_Stability-Plane}.
  \item It follows easily from the above.
  \end{enumerate}
\end{proof}

\begin{remark}
  Those results from \cite{Point_Vortex_Stability-Plane} mentioned above are for the relative dynamics of point vortices on the plane \textit{with} the translational invariance in addition to the rotational invariance.
  So the relative dynamics in the present paper is slightly different from that from \cite{Point_Vortex_Stability-Plane}.
  However they are both Lie--Poisson dynamics (with different Lie algebras with similar structures), and so a similar argument applies to the present setting.
\end{remark}

Since $\mu$ has $N^{2}$ real components and $R$ gives $(N-1)^{2}$ real components, the invariant set $R^{-1}(0)$ has the dimension $N^{2} - (N-1)^{2} = 2N - 1$---fewer than the original dimension $2N$.
Additionally, if the Casimir $C$ is independent of $R$, the effective dimension of the relative dynamics is $2N - 2$, which is what one would expect from the symplectic reduction theory~\cite{MaWe1974} in the presence of the $\mathbb{S}^{1}$-symmetry.

\subsection{Relative Equilibria}
The relative dynamics governed by the Lie--Poisson equation~\eqref{eq:LP} describes the time evolution of the shape made by the vortices and the origin regardless of its rotational orientation; see \Cref{fig:Symmetry}.
This implies that a fixed point in the Lie--Poisson equation~\eqref{eq:LP} corresponds to a \textit{relative equilibrium}, i.e., a solution to the original $N$-vortex equation~\eqref{eq:basic} in which the vortices undergo a rigid rotation about the origin without changing its relative configurations.

Therefore, one may analyze the stability of a relative equilibrium of \eqref{eq:basic} by analyzing the stability of the corresponding \textit{fixed point} in the Lie--Poisson equation~\eqref{eq:LP}.
We shall show a few examples of such applications in \Cref{sec:Applications} below.

\subsection{Nonlinear Stability of Relative Equilibria}
\label{ssec:nonlinear_stability}
For the Lyapunov stability, we may adapt the Energy--Casimir-type method proved in \citet[Theorem~5.2]{Point_Vortex_Stability-Plane} to the current setting as follows:

Let us first note that we shall identify $\u(N)_{\Gamma}^{*}$ with $\R^{N^{2}}$ in what follows.
Let $\mu_{0} \in R^{-1}(0)$ be a fixed point of the relative dynamics~\eqref{eq:LP}, and assume that $C$ and $R$ are independent at $\mu_{0}$.
Suppose that there exist $a_{0} \in \R\backslash\{0\}$, $a_{1} \in \R$, $\{ b_{i} \in \R \}_{i=1}^{N-1}$, and $\{ (c_{ij}, d_{ij}) \in \R^{2} \}_{1\le i < j \le N-1}$ such that
\begin{align*}
  % f\colon \vK \cong \R^{(N-1)^{2}} \to \R;
%    \qquad
  f(\mu) &\defeq a_{0} h(\mu) + a_{1} C(\mu) + \sum_{i=1}^{N-1} b_{i} R_{i}(\mu) \\
         &\qquad + \sum_{1\le i < j \le N-1} (c_{ij} \Re R_{ij}(\mu) + d_{ij} \Im R_{ij}(\mu))
\end{align*}
satisfies the following:
\begin{enumerate}[(i)]
\item $Df(\mu_{0}) = 0$; and
  \smallskip
\item the Hessian $\mathsf{H} \defeq D^{2}f(\mu_{0})$ is positive definite on the tangent space $T_{\mu_{0}}M$ at $\mu_{0}$ of the level set
  \begin{align*}
    M &\defeq \setdef{ \mu \in \u(N)_{\Gamma}^{*} \cong \R^{N^{2}} }{ R(\mu) = 0,\ C(\mu) = C(\mu_{0}) } \\
      &= R^{-1}(0) \cap C^{-1}(C(\mu_{0}))
  \end{align*}
  i.e., $v^{T} \mathsf{H} v > 0$ for every $v \in \R^{N^{2}} \backslash\{0\}$ such that $v \in \ker DC(\mu_{0}) \cap \ker DR(\mu_{0})$.
\end{enumerate}
Then $\mu_{0}$ is Lyapunov stable.

\begin{remark}
  The above criteria are a sufficient condition for the Hamiltonian $h$ to have a local minimum at $\mu_{0}$ in the level set $M$.
\end{remark}

\section{Applications}
\label{sec:Applications}
We consider applications of the above relative dynamics to those cases with $N = 2, 3, 4$.
The main application is to find relative equilibria and analyze their stability.
As mentioned above, we shall do so by finding fixed points in the Lie--Poisson equation~\eqref{eq:LP} and analyzing their stability as fixed points.

\subsection{Relative Equilibria and Stability with $N = 2$}
\label{ssec:Dumbbell}
Since relative equilibria in the vortex dipole case ($N = 2$ and $\Gamma_{1} = -\Gamma_{2}$) are studied in detail by \citet{ToKeFrCaScHa2011}, we shall consider the same-sign case with $\Gamma_{1} = \Gamma_{2}$, which is studied numerically in \citet{MuGrKuSi2016}.

\subsubsection{Relative Equilibria}
We shall identify $\u(2)_{\Gamma}^{*}$ with $\R^{4}$ using the coordinates $(\mu_{1}, \mu_{2}, \mu_{3}, \mu_{4})$ from \eqref{eq:mu-N=2}, that is,
\begin{equation*}
  \mu = \rmi
  \begin{bmatrix}
    \mu_{1} & \mu_{12} \\
    \mu_{12}^{*} & \mu_{2}
  \end{bmatrix}
  = \rmi
  \begin{bmatrix}
    \mu_{1} & \mu_{3} + \rmi \mu_{4} \\
    \mu_{3} - \rmi \mu_{4} & \mu_{2}
  \end{bmatrix}.
\end{equation*}
Then the Lie--Poisson equation~\eqref{eq:LP} gives
\begin{gather*}
  \dot{\mu}_{1} = 2c \Gamma_{1} \frac{\mu_{4}}{\mu_{1} + \mu_{2} - 2\mu_{3}},
  \quad
  \dot{\mu}_{2} = -2c \Gamma_{1} \frac{\mu_{4}}{\mu_{1} + \mu_{2} - 2\mu_{3}},
  \quad
  \dot{\mu}_{3} = \Gamma_{1} \frac{(\mu_{1} - \mu_{2})\mu_{4}}{(1 - \mu_{1})(1 - \mu_{2})},
  \\
  \dot{\mu}_{4} = -\Gamma_{1} \parentheses{
    c \frac{\mu_{1} - \mu_{2}}{\mu_{1} + \mu_{2} - 2\mu_{3}}
    + \parentheses{
      \frac{1}{1 - \mu_{1}} - \frac{1}{1 - \mu_{2}}
    } \mu_{3}
  }.
\end{gather*}
Therefore, $(\mu_{1}, \mu_{2}, \mu_{3}, \mu_{4})$ is a fixed point of \eqref{eq:LP} if and only if
\begin{equation}
  \label{eq:fixedpt-condition_N=2}
  c \frac{\mu_{1} - \mu_{2}}{\mu_{1} + \mu_{2} - 2\mu_{3}}
  + \parentheses{
    \frac{1}{1 - \mu_{1}} - \frac{1}{1 - \mu_{2}}
  } \mu_{3} = 0
  \quad\text{and}\quad
  \mu_{4} = 0.
\end{equation}

Recall from \eqref{eq:mu-shapevars_N=2} (see also \Cref{fig:Dumbbell}) that one may write
\begin{equation*}
  \mu_{1} = r_{1}^{2},
  \qquad
  \mu_{2} = r_{2}^{2},
  \qquad
  \mu_{3} = r_{1} r_{2} \cos\theta,
  \qquad
  \mu_{4} = r_{1} r_{2} \sin\theta.
\end{equation*}
The second condition $\mu_{4} = 0$ then implies that the fixed points necessarily take the form
\begin{equation}
  \label{eq:DumbbellRelEq}
  (\mu_{1}, \mu_{2}, \mu_{3}, \mu_{4}) = \left( r_{1}^{2}, r_{2}^{2}, \pm r_{1} r_{2}, 0 \right),
\end{equation}
including the special case where either $r_{1}$ or $r_{2}$ vanishes.
In fact, $\mu_{4} = r_{1} r_{2} \sin\theta = 0$ implies either (i)~$r_{i} \neq 0$ with $i = 1,2$ and $\sin\theta = 0$, or (ii)~$r_{1} r_{2} = 0$.
In the first case, $\cos\theta = \pm1$, and hence we have $\mu_{3} = r_{1} r_{2} \cos\theta = \pm r_{1} r_{2}$ as shown above.
In the second case, $\mu_{3} = r_{1} r_{2} \cos\theta = 0$, but then since $r_{1} r_{2} = 0$, one may write $\mu_{3} = \pm r_{1} r_{2}$ in this case as well.
Therefore, the the expression in \eqref{eq:DumbbellRelEq} applies to both cases.

Let us first consider the case with $\mu_{3} = r_{1} r_{2}$, i.e., the two vortices are both on a half line emanating from the origin to infinity.
Given that $r_{1}, r_{2} \in [0,1)$ and $r_{1} \neq r_{2}$ (no collisions), the first equation from \eqref{eq:fixedpt-condition_N=2} is then equivalent to
\begin{equation*}
  c (1 - r_{1}^{2}) (1 - r_{2}^{2}) = - r_{1} r_{2} (r_{1} - r_{2})^{2},
\end{equation*}
but then this is impossible since $r_{1}, r_{2} \in [0,1)$ and $c > 0$.
Hence there is no fixed point with $\mu_{3} = r_{1} r_{2}$.

Next consider the other case with $\mu_{3} = -r_{1} r_{2}$, i.e., the two vortices are on a line passing through the origin and are on the opposite sides of the origin; see \Cref{fig:DumbbellRelEq}.
\begin{figure}[htbp]
  \centering
  \includegraphics[width=.35\linewidth]{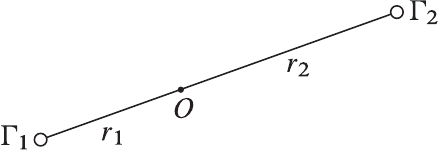}
  \caption{
    Relative equilibria for two same-sign vortices.
    The vortices are on a single line passing through the origin and are on the opposite sides of the origin; it is a relative equilibrium if $r_{1}$ and $r_{2}$ satisfy \eqref{eq:DumbbellCond}. 
  }
  \label{fig:DumbbellRelEq}
\end{figure}
A similar calculation assuming $r_{1} + r_{2} > 0$ (no collisions at the origin) shows that the first equation from \eqref{eq:fixedpt-condition_N=2} is equivalent to
\begin{equation}
  \label{eq:DumbbellCond}
  r_{1} = r_{2}
  \quad\text{or}\quad
  c (1 - r_{1}^{2}) (1 - r_{2}^{2}) = r_{1} r_{2} (r_{1} - r_{2})^{2}.
\end{equation}
See \Cref{fig:DumbbellCond} for the set of points in $(r_{1},r_{2})$-plane satisfying those conditions.

\begin{figure}[htbp]
  \centering
  \includegraphics[width=.35\linewidth]{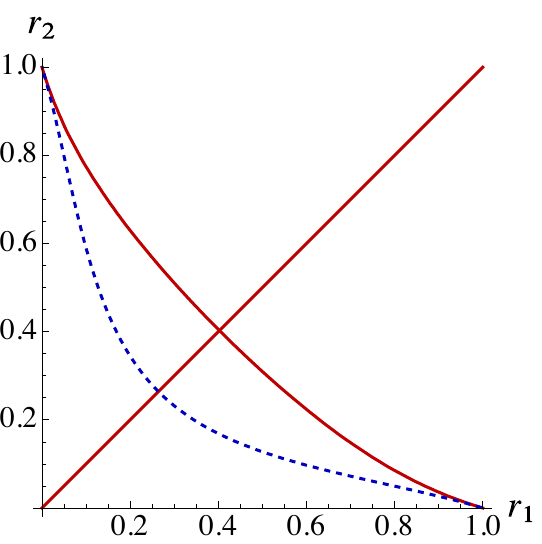}
  \caption{
    The solid red curves indicate the set of pairs of $(r_{1},r_{2})$ satisfying \eqref{eq:DumbbellCond}, under which \eqref{eq:DumbbellRelEq} gives relative equilibria for the same-sign case $\Gamma_{1} = \Gamma_{2}$.
    The dashed blue curve shows the same thing ($c (1 - r_{1}^{2}) (1 - r_{2}^{2}) = r_{1} r_{2} (2 - r_{1}^{2} - r_{2}^{2})$ from \cite[Eq.~(23)]{ToKeFrCaScHa2011}) for the dipole case $\Gamma_{1} = -\Gamma_{2}$.
    Both with $c = 0.15$.
  }
  \label{fig:DumbbellCond}
\end{figure}

Unlike the previous case, there are infinitely many pairs $(r_{1}, r_{2})$ satisfying the second equation.
The second equation is similar to the corresponding equation for the dipole case, while the dipole case does not have fixed points with $r_{1} = r_{2}$; see \cite[Eq.~(23)]{ToKeFrCaScHa2011}.

\subsubsection{Stability of Relative Equilibria}
We may analyze the stability of the above equilibria using both linear and nonlinear stability analysis.

\begin{proposition}
  \label{prop:stability-Dumbbell}
  Consider the relative equilibria for a pair of vortices of the same sign found above (as in \Cref{fig:DumbbellCond}):
  \begin{equation*}
    \mu_{0} = \left( r_{1}^{2}, r_{2}^{2}, - r_{1} r_{2}, 0 \right) \text{ with }
    \left\{
      \begin{array}{l}
        A:~r_{1} = r_{2} \text{ or}\smallskip\\
        B:~c (1 - r_{1}^{2}) (1 - r_{2}^{2}) = r_{1} r_{2} (r_{1} - r_{2})^{2}.
      \end{array}
    \right.
  \end{equation*}
  \begin{enumerate}[(i)]
  \item Relative equilibrium~A is Lyapunov stable if $c(1 - r_{1}^{2})^{2} > 4r_{1}^{4}$ and unstable if $c(1 - r_{1}^{2})^{2} < 4r_{1}^{4}$.
  \item Defining
    \begin{equation}
      \label{eq:F}
      F_{1}(r_{1},r_{2}) \defeq \left(r_{1}^{2} + r_{2}^{2}\right)^{2} + \left(r_{1}^{2} + r_{2}^{2}\right)\left(r_{1}^{2}r_{2}^{2} + 4 r_{1} r_{2} - 3\right) + 2(r_{1} - r_{2})^{2},
    \end{equation}
    relative equilibria~B is Lyapunov stable if $F_{1}(r_{1},r_{2}) < 0$ and unstable if $F_{1}(r_{1},r_{2}) > 0$.
  \end{enumerate}
  See \Cref{fig:DumbbellStabCond} for the stable and unstable conditions.
\end{proposition}
\begin{figure}[htbp]
  \centering
  \begin{subcaptionblock}[c]{0.45\textwidth}
    \centering
    \includegraphics[width=.9\linewidth]{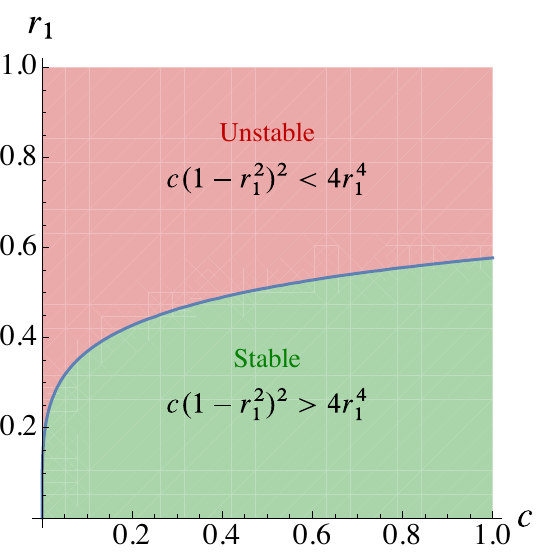}
    \caption{Relative equilibrium~A}
    \label{fig:DumbbellStabCond1}
  \end{subcaptionblock}
  \hfill
  \begin{subcaptionblock}[c]{0.45\textwidth}
    \centering
    \includegraphics[width=.9\linewidth]{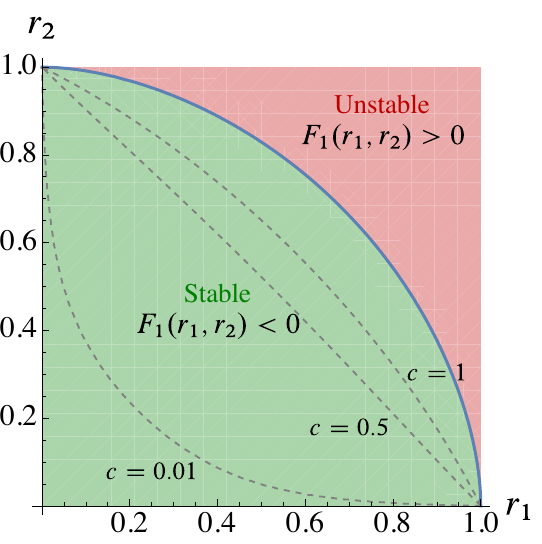}
    \caption{Relative equilibrium~B}
    \label{fig:DumbbellStabCond2}
  \end{subcaptionblock}
  \caption{
    Stability conditions from \Cref{prop:stability-Dumbbell}.
    \subref{fig:DumbbellStabCond1}~Stable and unstable domains in $(c,r_{1})$-plane of relative equilibrium~A.
    One sees that, for each value of $c$, there is an upper threshold for $r_{1}$ so that the relative equilibrium~A is stable.
    \subref{fig:DumbbellStabCond2}~The dashed curves are the pairs $(r_{1},r_{2})$ satisfying the condition $c (1 - r_{1}^{2}) (1 - r_{2}^{2}) = r_{1} r_{2} (r_{1} - r_{2})^{2}$ for relative equilibrium~B for $c = 0.01, 0.5, 0.1$.
    The blue solid curve is the contour $F_{1}(r_{1},r_{2}) = 0$.
    One sees that, regardless of the value of $c$, every relative equilibrium~B is stable.
    }
  \label{fig:DumbbellStabCond}
\end{figure}
\begin{proof}[Proof or \Cref{prop:stability-Dumbbell}]
  \begin{enumerate}[(i)]
  \item Let us first show the instability by linear stability analysis.
    Linearizing the Lie--Poisson equation \eqref{eq:LP} at relative equilibrium A gives a linear system in $\R^{4}$ with the matrix
    \begin{equation*}
      \frac{\Gamma_{1}}{4 r_{1}^{2} (1 - r_{1}^{2})^{2}}
      \begin{bmatrix}
        0 & 0 & 0 & g_{1} \\
        0 & 0 & 0 & -g_{1} \\
        0 & 0 & 0 & 0 \\
        g_{2} & -g_{2} & 0 & 0
      \end{bmatrix}
      \text{ with }
      \left\{
        \begin{array}{l}
          g_{1} \defeq 2c (1 - r_{1}^{2})^{2}, \medskip\\
          g_{2} \defeq 4r_{1}^{4} - c(1 - r_{1}^{2})^{2}.
        \end{array}
      \right.
    \end{equation*}
    Its eigenvalues are
    \begin{equation*}
      0,\, 0,\, \pm \frac{\sqrt{c}\,\Gamma_{1}}{2r_{1}^{2}(1 - r_{1}^{2})}\,\sqrt{g_{2}},
    \end{equation*}
    and so relative equilibria~A is unstable if $g_{2} > 0$.

    For the Lyapunov stability, we shall use the Energy--Casimir-type method from \Cref{ssec:nonlinear_stability}:
    Stetting
    \begin{equation*}
      f(\mu) = a_{0} h(\mu) + a_{1} C(\mu) + b_{1} R(\mu),
    \end{equation*}
    where $h$ is given in \eqref{eq:h-N=2} (with $\Gamma_{2} = \Gamma_{1}$ here) and
    \begin{equation*}
      C(\mu) = \Gamma_{1} (\mu_{1} + \mu_{2}),
      \qquad
      R(\mu) = \det(-\rmi\mu) = \mu_{1} \mu_{2} - \mu_{3}^{2} - \mu_{4}^{2}.
    \end{equation*}
    Then one sees that $Df(\mu_{0}) = 0$ if
    \begin{equation*}
      a_{1} =  \Gamma_{1} \frac{c(1 - r_{1}^{2}) + 2r_{1}^{2}}{4 r_{1}^{2}(1 - r_{1}^{2})}\, a_{0},
      \qquad
      b_{1} = -\Gamma_{1}^{2} \frac{c}{8 r_{1}^{4}}\, a_{0}.
    \end{equation*}
    In order to find the tangent space to the level set $M = R^{-1}(0) \cap C^{-1}(C(\mu_{0}))$, we compute
    \begin{equation*}
      DC(\mu) = \Gamma_{1}
      \begin{bmatrix}
        1 \\
        1 \\
        0 \\
        0
      \end{bmatrix},
      \qquad
      DR(\mu) =
      \begin{bmatrix}
        \mu_{2} \\
        \mu_{1} \\
        -2\mu_{3} \\
        -2\mu_{4}
      \end{bmatrix}
      \implies
      DR(\mu_{0}) = r_{1}^{2}
      \begin{bmatrix}
        1 \\
        1 \\
        -2 \\
        0
      \end{bmatrix}.
    \end{equation*}
    Then the tangent space $T_{\mu_{0}}M$ is given by
    \begin{equation*}
      T_{\mu_{0}}M = \ker DC(\mu_{0}) \cap \ker DR(\mu_{0})
      = \Span\braces{
        v_{1} \defeq 
        \begin{bmatrix}
          1 \\
          -1 \\
          0 \\
          0
        \end{bmatrix},\,
        v_{2} \defeq \begin{bmatrix}
          0 \\
          0 \\
          0 \\
          1
        \end{bmatrix}
      }.
    \end{equation*}
    Then, defining the $2 \times 2$ matrix $\mathcal{H}$ by setting $\mathcal{H}_{ij} \defeq v_{i}^{T} D^{2}f(\mu_{0}) v_{j}$, the matrix $\mathcal{H}$ is the diagonal matrix with diagonal entries
    \begin{align*}
      \frac{c\,\Gamma_{1}^{2}}{4 r_{1}^{4}}\, a_{0}
      \quad\text{and}\quad
      -\frac{c\,\Gamma_{1}^{4}}{16 r_{1}^{8} (1 - r_{1}^{2})^{2}}\, a_{0}^{2}\,g_{2}.
    \end{align*}
    Since $c > 0$, $\Gamma_{1} \neq 0$, and $0 < r_{1} < 1$, we take an arbitrary $a_{0} > 0$; then both become positive if $g_{2} < 0$.
  \item Proceeding the same way as above, the linearization at relative equilibrium~B yields a linear system in $\R^{4}$ with eigenvalues
    \begin{equation*}
      0,\, 0,\, \pm \frac{(r_{1} - r_{2})\Gamma_{1}}{(1 - r_{1}^{2})^{3/2}(1 - r_{2}^{2})^{3/2}}\,\sqrt{F_{1}(r_{1},r_{2})}
    \end{equation*}
    using the function $F_{1}$ defined in \eqref{eq:F}.

    The nonlinear analysis proceeds in a similar way too.
    One sees that $Df(\mu_{0}) = 0$ if
    \begin{equation*}
      a_{1} = \Gamma_{1} \frac{(1 + r_{1}r_{2})\Gamma_{1}}{2(1 - r_{1}^{2})(1 - r_{2}^{2})}\, a_{0},
      \qquad
      b_{1} = - \frac{\Gamma_{1}^{2}}{2(1 - r_{1}^{2})(1 - r_{2}^{2})}\, a_{0}.
    \end{equation*}
    We also see that
    \begin{equation*}
      DR(\mu_{0}) = 
      \begin{bmatrix}
        r_{2}^{2} \\
        r_{1}^{2} \\
        2 r_{1} r_{2} \\
        0
      \end{bmatrix},
    \end{equation*}
    and thus
    \begin{equation*}
      T_{\mu_{0}}M = \ker DC(\mu_{0}) \cap \ker DR(\mu_{0})
      = \Span\braces{
        v_{1} \defeq
        \frac{1}{r_{1}^{2} - r_{2}^{2}}
        \begin{bmatrix}
          2 r_{1} r_{2} \\
          -2 r_{1} r_{2} \\
          r_{1}^{2} - r_{2}^{2} \\
          0
        \end{bmatrix},\,
        v_{2} \defeq \begin{bmatrix}
          0 \\
          0 \\
          0 \\
          1
        \end{bmatrix}
      }.
    \end{equation*}
    Then $\mathcal{H}$ is the diagonal matrix with diagonal entries
    \begin{equation*}
      \frac{\Gamma_{1}^{2}}{(1 - r_{1}^{2})(1 - r_{2}^{2})}\,a_{0},
      \quad\text{and}\quad
      -\frac{\Gamma_{1}^{4}}{(r_{1} + r_{2})^{2}(1 - r_{1}^{2})^{3}(1 - r_{2}^{2})^{3}}\,a_{0}^{2}\, F_{1}(r_{1}, r_{2}).
    \end{equation*}
    Taking an arbitrary $a_{0} > 0$, we see that both become positive if $F_{1}(r_{1}, r_{2}) < 0$.
  \end{enumerate}
\end{proof}

\subsection{Relative Equilibria and Stability with $N = 3$}
\subsubsection{Equilateral Relative Equilibrium}
As an example of a relative equilibrium with $N = 3$, consider the equilateral relative equilibrium of three vortices; see \Cref{fig:EquilateralRelEq}.
\begin{figure}[htbp]
  \centering
  \includegraphics[width=.35\linewidth]{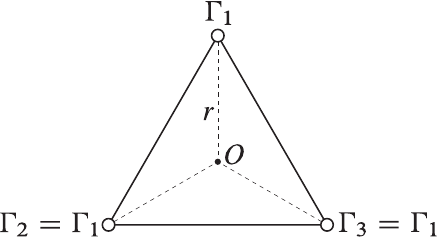}
  \caption{
    The equilateral configuration with the origin at the center is a relative equilibrium if and only if $\Gamma_{1} = \Gamma_{2} = \Gamma_{3}$.
  }
  \label{fig:EquilateralRelEq}
\end{figure}

For $N = 3$, we may use coordinates $(\mu_{1}, \dots, \mu_{9}) \in \R^{9}$ to write
\begin{equation*}
  \mu = \rmi
    \begin{bmatrix}
      \mu_{1} & \mu_{12} & \mu_{13} \\
      \mu_{12}^{*} & \mu_{2} & \mu_{23} \\
      \mu_{13}^{*} & \mu_{23}^{*} & \mu_{3}
    \end{bmatrix}
    = \rmi
    \begin{bmatrix}
      \mu_{1} & \mu_{4} + \rmi \mu_{5} & \mu_{6} + \rmi\mu_{7} \\
      \mu_{4} - \rmi \mu_{5} & \mu_{2} & \mu_{8} + \rmi\mu_{9} \\
      \mu_{6} - \rmi\mu_{7} & \mu_{8} - \rmi\mu_{9} & \mu_{3} \\
    \end{bmatrix}.
\end{equation*}
Then the equilateral relative equilibrium pictured in \Cref{fig:EquilateralRelEq} (with $\Gamma_{1} = \Gamma_{2} = \Gamma_{3}$) corresponds to the fixed point $\mu = \mu_{0}$ of the Lie--Poisson equation~\eqref{eq:LP} in which
\begin{equation*}
  \mu_{i} = |z_{i}|^{2} = r^{2} \quad \forall i \in \{1, 2, 3\},
  \qquad
  \mu_{12}^{*} = \mu_{23}^{*} = \mu_{13} = z_{1} z_{3}^{*} = r^{2} e^{\rmi(2\pi/3)}
\end{equation*}
or equivalently, in terms of $(\mu_{1}, \dots, \mu_{9}) \in \R^{9}$,
\begin{equation*}
  \mu_{0} =
  \parentheses{ r^{2}, r^{2}, r^{2}, -\frac{r^{2}}{2}, -\frac{\sqrt{3}}{2} r^{2}, -\frac{r^{2}}{2}, \frac{\sqrt{3}}{2} r^{2}, -\frac{r^{2}}{2}, -\frac{\sqrt{3}}{2} r^{2} }.
\end{equation*}

\subsubsection{Stability of Equilateral Relative Equilibrium}
Since the linear stability analysis is performed for the more general $N$-ring cases in \citet{KoKeCa2014}, we shall only briefly touch on our linear analysis just to confirm that it reproduces the same result.
Indeed, the linearization of the Lie--Poisson equation~\eqref{eq:LP} at the above equilibrium $\mu_{0}$ yields a linear system with eigenvalues
\begin{equation*}
  0,\, 0,\, 0,\,
  \pm\rmi\frac{c\,\Gamma_{1}}{r^{2}},
  \,
  \pm\frac{\sqrt{c}\,\Gamma_{1}}{r^{2}(1 - r^{2})} \sqrt{ F_{2}(c,r) }\,
\end{equation*}
with
\begin{equation*}
  F_{2}(c,r) \defeq 2r^{4} - c(1 - r^{2})^{2},
\end{equation*}
and each of the last pair of conjugate eigenvalues has algebraic multiplicity 2.
This indicates that the fixed point is unstable (and hence so is the relative equilibrium) if $F_{2}(c,r) > 0$, as in \cite[Theorem~3.1]{KoKeCa2014} for $N = 3$.

It is also shown in \cite[Theorem~3.1]{KoKeCa2014} that the relative equilibrium is stable if $F_{2}(c,r) < 0$ by linear analysis.
We shall perform a nonlinear stability analysis using the Energy--Casimir-type method from \Cref{ssec:nonlinear_stability} to show that the fixed point is Lyapunov stable.
Specifically, we have
\begin{equation*}
  f(\mu) = a_{0} h(\mu) + a_{1} C(\mu) + \sum_{i=1}^{2} b_{i} R_{i}(\mu) + c_{12} \Re R_{12}(\mu) + d_{12} \Im R_{12}(\mu),
\end{equation*}
where $h$ is given in \eqref{eq:h} and
\begin{gather*}
  C(\mu) = \Gamma_{1} (\mu_{1} + \mu_{2} + \mu_{3}),
  \\
  R_{1}(\mu) = \mu_{1} \mu_{2} - \mu_{4}^{2} - \mu_{5}^{2},
  \qquad
  R_{2}(\mu) = \mu_{2} \mu_{3} - \mu_{8}^{2} - \mu_{9}^{2},
  \\
  R_{12}(\mu) =
  \begin{vmatrix}
    \mu_{4} + \rmi \mu_{5} & \mu_{6} + \rmi\mu_{7} \\
    \mu_{2} & \mu_{8} + \rmi\mu_{9}
  \end{vmatrix}.
\end{gather*}
Then one sees that $Df(\mu_{0}) = 0$ if
\begin{equation*}
  a_{1} = \Gamma_{1} \frac{c(1 - r^{2}) + r^{2}}{2 r^{2}(1 - r^{2})}\, a_{0},
  \qquad
  b_{1} = b_{2} = -\Gamma_{1}^{2} \frac{c}{6 r^{4}}\, a_{0},
  \qquad
  c_{12} = \Gamma_{1}^{2} \frac{c}{3 r^{4}}\, a_{0},
  \qquad
  d_{12} = 0.
\end{equation*}
Now, writing $R = (R_{1}, R_{2}, \Re R_{12}, \Im R_{12})$ and setting $M \defeq R^{-1}(0) \cap C^{-1}(C(\mu_{0}))$, it is straightforward calculations to see that a basis for $T_{\mu_{0}}M$ is given by
\begin{gather*}
  v_{1} \defeq \sqrt{3} (e_{1} - e_{3}) - e_{5} + e_{9},
  \quad
  v_{2} \defeq e_{1} - e_{3} - e_{4} + e_{8},
  \\
  v_{3} \defeq \sqrt{3} (-e_{2} + e_{3}) + e_{5} + e_{7},
  \quad
  v_{4} \defeq e_{2} - e_{3} - e_{4} + e_{6}.
\end{gather*}
using the standard basis $\{ e_{i} \}_{i=1}^{9}$ for $\R^{9}$.
Then, defining the $4 \times 4$ matrix $\mathcal{H}$ by setting $\mathcal{H}_{ij} \defeq r^{4} (1 - r^{2})^{2}\, v_{i}^{T} D^{2}f(\mu_{0}) v_{j}$ (the factors $r^{4} (1 - r^{2})^{2}$ are multiplied to simplify the expression), its leading principal minors are
\begin{gather*}
  -\frac{a_{0} \Gamma_{1}^{2}}{3}\, F_{3}(c,r),
  \quad
  -\frac{4a_{0}^{2} \Gamma_{1}^{4} c}{3} (1 - r^{2})^{2} F_{2}(c,r),
  \\
  \frac{a_{0}^{3} \Gamma_{1}^{6} c}{3} (1 - r^{2})^{2} F_{2}(c,r) F_{3}(c,r),
  \quad
  a_{0}^{4} \Gamma_{1}^{8} c^{2} (1 - r^{2})^{4} F_{2}(c,r)^{2},
\end{gather*}
where
\begin{equation*}
  F_{3}(c,r) \defeq 9 r^{4} - 5 c (1 - r^{2})^{2} = \frac{9}{2} F_{2}(c,r) - \frac{1}{2} c (1 - r^{2})^{2}.
\end{equation*}
Since $0 < r < 1$ and $c > 0$, all the leading principal minors are positive if $F_{2}(c,r) < 0$ and $F_{3}(c,r) < 0$ by choosing an arbitrary $a_{0} > 0$.
However, since $F_{3}(c,r) < \frac{9}{2}F_{2}(c,r)$ as shown above, $F_{2}(c,r) < 0$ is sufficient.
Hence we have shown the Lyapunov stability under the same condition for the linear stability from \cite[Theorem~3.1]{KoKeCa2014} with $N = 3$.

\subsection{Relative Equilibria and Stability with $N = 4$}
\subsubsection{Equilateral with Center}
As an example with $N = 4$, consider a slight variant of the above by adding another vortex at the center.
This configuration also gives a relative equilibrium if $\Gamma_{1} = \Gamma_{2} = \Gamma_{3}$ regardless of the value of the topological charge $\Gamma_{4}$ of the center vortex; see \Cref{fig:EquilateralWithCenterRelEq}.
\begin{figure}[htbp]
  \centering
  \includegraphics[width=.35\linewidth]{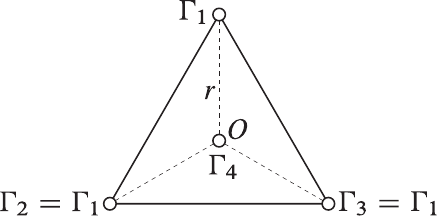}
  \caption{
    Relative equilibrium of equilateral triangle with center:
    Three vortices of equal topological charge ($\Gamma_{1} = \Gamma_{2} = \Gamma_{3}$) are at the vertices of an equilateral triangle with the center at the origin, and another vortex with an arbitrary topological charge $\Gamma_{4}$ at the center.
  }
  \label{fig:EquilateralWithCenterRelEq}
\end{figure}

For $N = 4$, we may use coordinates $(\mu_{1}, \dots, \mu_{16}) \in \R^{16}$ to write
\begin{equation*}
  \mu = \rmi
    \begin{bmatrix}
      \mu_{1} & \mu_{12} & \mu_{13} & \mu_{14}\\
      \mu_{12}^{*} & \mu_{2} & \mu_{23} & \mu_{24}\\
      \mu_{13}^{*} & \mu_{23}^{*} & \mu_{3} & \mu_{34} \\
      \mu_{14}^{*} & \mu_{24}^{*} & \mu_{34}^{*} & \mu_{4}
    \end{bmatrix}
    = \rmi
    \begin{bmatrix}
      \mu_{1} & \mu_{5} + \rmi \mu_{6} & \mu_{7} + \rmi\mu_{8} & \mu_{9} + \rmi\mu_{10}\\
      \mu_{5} - \rmi \mu_{6} & \mu_{2} & \mu_{11} + \rmi\mu_{12} & \mu_{13} + \rmi\mu_{14} \\
      \mu_{7} - \rmi\mu_{8} & \mu_{11} - \rmi\mu_{12} & \mu_{3} & \mu_{15} + \rmi\mu_{16} \\
      \mu_{9} - \rmi\mu_{10} & \mu_{13} - \rmi\mu_{14} & \mu_{15} - \rmi\mu_{16} & \mu_{4}
    \end{bmatrix},
\end{equation*}
Then the relative equilibrium in question corresponds to the fixed point $\mu = \mu_{0}$ of the Lie--Poisson equation~\eqref{eq:LP} in which
\begin{gather*}
  \mu_{i} = |z_{i}|^{2} = r^{2} \quad \forall i \in \{1, 2, 3\},
  \qquad
  \mu_{4} = 0,
  \\
  \mu_{12}^{*} = \mu_{23}^{*} = \mu_{13} = z_{1} z_{3}^{*} = r^{2} e^{\rmi(2\pi/3)},
  \qquad
  \mu_{i4} = 0 \quad \forall i \in \{1, 2, 3\},
\end{gather*}
or equivalently, in terms of $(\mu_{1}, \dots, \mu_{16}) \in \R^{16}$,
\begin{equation*}
  \mu_{0} =
  \parentheses{ r^{2}, r^{2}, r^{2}, 0, -\frac{r^{2}}{2}, -\frac{\sqrt{3}}{2} r^{2}, -\frac{r^{2}}{2}, \frac{\sqrt{3}}{2} r^{2}, 0, 0,
    -\frac{r^{2}}{2}, -\frac{\sqrt{3}}{2} r^{2}, 0, 0, 0, 0 }.
\end{equation*}

\subsubsection{Stability of Equilateral with Center}
We shall not perform the linear stability analysis here because the characteristic equation for the eigenvalues of the linearized system becomes very complicated due to the high-dimensionality of the problem.

On the other hand, the nonlinear analysis is more tractable thanks to the effective dimension reduction using the constraint to the system:
\begin{proposition}
  The relative equilibrium of equilateral triangle with center (see \Cref{fig:EquilateralWithCenterRelEq}) with $\Gamma_{1} = \Gamma_{2} = \Gamma_{3} = \Gamma_{4} = \pm1$ is stable if the parameters $c$ and $r$ satisfy (see \Cref{fig:EquilateralWithCenterStabCond})
  \begin{equation*}
    F_{6}(c,r) \defeq 2 r^{4} - c(1 - r^{2})(2 - 3r^{2}) < 0.
  \end{equation*}
\end{proposition}
\begin{figure}[htbp]
  \centering
  \includegraphics[width=.375\linewidth]{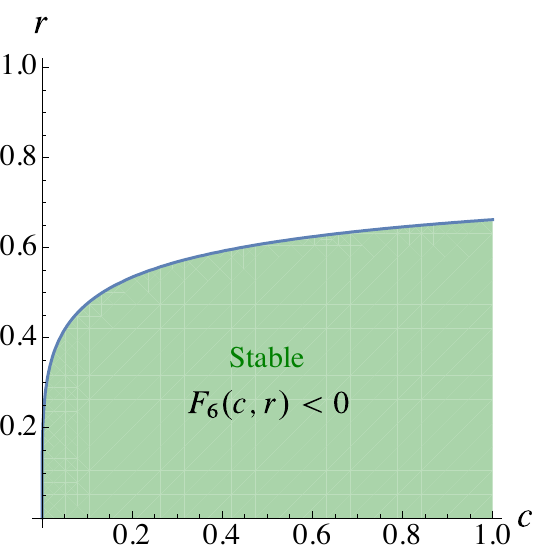}
  \caption{
    The relative equilibrium of the equilateral with center from \Cref{fig:EquilateralWithCenterRelEq} with $\Gamma_{1} = \Gamma_{2} = \Gamma_{3} = \Gamma_{4} = \pm1$ is stable if $F_{6}(c,r) < 0$, the green domain in the $(c,r)$-plane below the blue curve $F_{6}(c,r) = 0$.
  }
  \label{fig:EquilateralWithCenterStabCond}
\end{figure}
\begin{proof}
  Since the case with $\Gamma_{i} = -1$ with $1 \le i \le 4$ is essentially the same as the case with $\Gamma_{i} = 1$ with $1 \le i \le 4$, we shall only consider the latter case for simplicity.
  
  Our Lyapunov function takes the form
  \begin{equation*}
    f(\mu) = a_{0} h(\mu) + a_{1} C(\mu) + \sum_{i=1}^{3} b_{i} R_{i}(\mu) + \sum_{1\le i < j \le 3} (c_{ij} \Re R_{ij}(\mu) + d_{ij} \Im R_{ij}(\mu)),
  \end{equation*}
  where $h$ is given in \eqref{eq:h} and
  \begin{gather*}
    C(\mu) = \mu_{1} + \mu_{2} + \mu_{3} + \mu_{4},
    \\
    R_{1}(\mu) = \mu_{1} \mu_{2} - \mu_{5}^{2} - \mu_{6}^{2},
    \quad
    R_{2}(\mu) = \mu_{2} \mu_{3} - \mu_{11}^{2} - \mu_{12}^{2},
    \quad
    R_{3}(\mu) = \mu_{3} \mu_{4} - \mu_{15}^{2} - \mu_{16}^{2},
    \\
    R_{12}(\mu) =
    \begin{vmatrix}
      \mu_{5} + \rmi \mu_{6} & \mu_{7} + \rmi\mu_{8} \\
      \mu_{2} & \mu_{11} + \rmi\mu_{12}
    \end{vmatrix},
    \quad
    R_{13}(\mu) =
    \begin{vmatrix}
      \mu_{7} + \rmi \mu_{8} & \mu_{9} + \rmi\mu_{10} \\
      \mu_{11} + \rmi\mu_{12} & \mu_{12} + \rmi\mu_{14}
    \end{vmatrix},
    \\
    R_{23}(\mu) =
    \begin{vmatrix}
      \mu_{11} + \rmi\mu_{12} & \mu_{12} + \rmi\mu_{14} \\
      \mu_{3} & \mu_{15} + \rmi\mu_{16}
    \end{vmatrix}.
  \end{gather*}
  Then one sees that $Df(\mu_{0}) = 0$ if
  \begin{gather*}
    a_{1} = \parentheses{ \frac{1}{2(1 - r^{2})} + \frac{c}{r^{2}} } a_{0},
    \quad
    b_{1} = b_{2} = - \frac{c}{6 r^{4}}\, a_{0},
    \quad
    b_{3} = \frac{1}{2r^{4}} \parentheses{ c - \frac{r^{4}}{1 - r^{2}} } a_{0},
    \\
    c_{12} = \frac{c}{3 r^{4}}\, a_{0},
    \quad
    -c_{13} = c_{23} = \frac{c}{2r^{4}}\,a_{0},
    \quad
    d_{12} = 0,
    \quad
    -d_{13} = d_{23} = \frac{\sqrt{3}\,c}{2r^{4}}\,a_{0}.
  \end{gather*}

  Writing $R = (R_{1}, R_{2}, R_{3}, \Re R_{12}, \dots, \Im R_{23})$ and setting $M \defeq R^{-1}(0) \cap C^{-1}(C(\mu_{0}))$, one finds that a basis for $T_{\mu_{0}}M$ is given by
\begin{gather*}
  v_{1} \defeq \frac{\sqrt{3}}{2} (-e_{9} + e_{13}) - \frac{1}{2}(e_{10} + e_{14}) + e_{16},
  \\
  v_{2} \defeq - \frac{1}{2}(e_{9} + e_{13}) + \frac{\sqrt{3}}{2} (e_{10} - e_{14}) + e_{15},
  \\
  v_{3} \defeq \sqrt{3} (e_{1} - e_{3}) + (-e_{6} + e_{12}),
  \quad
  v_{4} \defeq e_{1} - e_{3} - e_{5} + e_{11},
  \\
  v_{5} = \sqrt{3}(-e_{2}  + e_{3}) + e_{6} + e_{8},
  \quad
  v_{6} = e_{2} - e_{3} - e_{5} + e_{7}
\end{gather*}
using the standard basis $\{ e_{i} \}_{i=1}^{16}$ for $\R^{16}$.

Then, defining the $6 \times 6$ matrix $\mathcal{H}$ by setting $\mathcal{H}_{ij} \defeq r^{4} (1 - r^{2})\, v_{i}^{T} D^{2}f(\mu_{0}) v_{j}$, its leading principal minors are
\begin{gather*}
  m_{1} \defeq a_{0}\, F_{4}(c,r),
  \quad
  m_{2} \defeq a_{0}^{2}\, F_{4}(c,r)^{2},
  \quad
  m_{3} \defeq -\frac{a_{0}^{3}}{3(1 - r^{2})} F_{4}(c,r) F_{5}(c,r),
  \\
  m_{4} \defeq -\frac{4}{3} a_{0}^{4}\, c\, r^{4} F_{4}(c,r) F_{6}(c,r),
  \quad
  m_{5} \defeq a_{0}^{5}\, \frac{c\,r^{4}}{3(1 - r^{2})} F_{5}(c,r) F_{6}(c,r),
  \quad
  m_{6} \defeq a_{0}^{6}\, c^{2}\, r^{8} F_{6}(c,r)^{2}.
\end{gather*}
where
\begin{align*}
  F_{4}(c,r) &\defeq r^{4} + 2c (1 - r^{2}),
  \\
  F_{5}(c,r) &\defeq 9 r^{8} + 2c\,r^{4}(1 - r^{2})(7r^{2} + 2) - 25 c^{2}(1 - r^{2})^{3},
  \\
  F_{6}(c,r) &\defeq 2 r^{4} - c(1 - r^{2})(2 - 3r^{2}).
\end{align*}

Given that $c > 0$ and $0 < r < 1$, we see that $F_{4}(c,r) > 0$, and hence $m_{2} > 0$.
One may then take an arbitrary $a_{0} > 0$ so that $m_{1} > 0$ as well.
Moreover, if both $F_{5}(c,r)$ and $F_{6}(c,r)$ are negative, $m_{i} > 0$ for every $i \ge 3$ as well; hence it implies the Lyapunov stability of the fixed point.

However, it turns out that $F_{6} < 0$ implies $F_{5} < 0$.
Indeed, the expression for $F_{6}$ implies that, if $F_{6} < 0$ then
\begin{equation*}
  0 < 2 r^{4} < c(1 - r^{2})(2 - 3r^{2}),
\end{equation*}
and so it is necessarily the case that $2 - 3r^{2} > 0$ because $0 < r < 1$.
Thus we see that
\begin{equation*}
  1 - r^{2} = \frac{1}{2}(2 - 2r^{2}) > \frac{1}{2}(2 - 3r^{2}) > 0.
\end{equation*}
Then one can bound $F_{5}$ above as follows:
\begin{align*}
  F_{5}(c,r) &< 9 r^{8} + 2c\,r^{4}(1 - r^{2})(7r^{2} + 2) - \frac{25}{2} c^{2}(1 - r^{2})^{2}(2 - 3r^{2}) \\
             &= \frac{1}{2}\left( 9r^{4} + 25c\,(1- r^{2}) \right) F_{6}(c,r) - \frac{1}{2} c\,r^{4} (1 - r^{2})(24 - r^{2}).
\end{align*}
This shows that $F_{6} < 0$ implies $F_{5} < 0$ as well, and hence the claimed result follows.
\end{proof}

\section*{Acknowledgments}
This work was supported by NSF grant DMS-2006736.

\bibliography{Confined_BEC_Vortices}

\begin{thebibliography}{27}
\providecommand{\natexlab}[1]{#1}
\providecommand{\url}[1]{\texttt{#1}}
\expandafter\ifx\csname urlstyle\endcsname\relax
  \providecommand{\doi}[1]{doi: #1}\else
  \providecommand{\doi}{doi: \begingroup \urlstyle{rm}\Url}\fi

\bibitem[Bolsinov et~al.(1999)Bolsinov, Borisov, and Mamaev]{BoBoMa1999}
A.~V. Bolsinov, A.~V. Borisov, and I.~S. Mamaev.
\newblock Lie algebras in vortex dynamics and celestial mechanics---{IV}.
\newblock \emph{Regular and Chaotic Dynamics}, 4\penalty0 (1):\penalty0 23--50,
  1999.

\bibitem[Borisov and Pavlov(1998)]{BoPa1998}
A.~V. Borisov and A.~E. Pavlov.
\newblock Dynamics and statics of vortices on a plane and a sphere---{I}.
\newblock \emph{Regular and Chaotic Dynamics}, 3\penalty0 (1):\penalty0 28--38,
  1998.

\bibitem[Fetter(2009)]{Fe2009}
A.~L. Fetter.
\newblock Rotating trapped {B}ose-{E}instein condensates.
\newblock \emph{Rev. Mod. Phys.}, 81:\penalty0 647--691, 2009.

\bibitem[Fetter and Svidzinsky(2001)]{FeSv2001}
A.~L. Fetter and A.~A. Svidzinsky.
\newblock Vortices in a trapped dilute {B}ose-{E}instein condensate.
\newblock \emph{Journal of Physics: Condensed Matter}, 13\penalty0
  (12):\penalty0 R135--R194, 2001.

\bibitem[Freilich et~al.(2010)Freilich, Bianchi, Kaufman, Langin, and
  Hall]{FrBiKaLaHa2010}
D.~V. Freilich, D.~M. Bianchi, A.~M. Kaufman, T.~K. Langin, and D.~S. Hall.
\newblock Real-time dynamics of single vortex lines and vortex dipoles in a
  {B}ose-{E}instein condensate.
\newblock \emph{Science}, 329\penalty0 (5996):\penalty0 1182--1185, 2010.

\bibitem[Goodman et~al.(2015)Goodman, Kevrekidis, and
  Carretero-Gonz\'{a}lez]{GoKeCa2015}
R.~H. Goodman, P.~G. Kevrekidis, and R.~Carretero-Gonz\'{a}lez.
\newblock Dynamics of vortex dipoles in anisotropic {B}ose--{E}instein
  condensates.
\newblock \emph{SIAM Journal on Applied Dynamical Systems}, 14\penalty0
  (2):\penalty0 699--729, 2015.

\bibitem[Guillemin and Sternberg(1980)]{GuSt1980}
V.~Guillemin and S.~Sternberg.
\newblock The moment map and collective motion.
\newblock \emph{Annals of Physics}, 127\penalty0 (1):\penalty0 220--253, 1980.

\bibitem[Guillemin and Sternberg(1990)]{GuSt1990}
V.~Guillemin and S.~Sternberg.
\newblock \emph{Symplectic Techniques in Physics}.
\newblock Cambridge University Press, 1990.

\bibitem[Kevrekidis et~al.(2004)Kevrekidis, Carretero-Gonz{\'a}lez,
  Frantzeskakis, and Kevrekidis]{KeCaFrKe2004}
P.~G. Kevrekidis, R.~Carretero-Gonz{\'a}lez, D.~J. Frantzeskakis, and I.~G.
  Kevrekidis.
\newblock Vortices in {B}ose--{E}instein condensates: Some recent developments.
\newblock \emph{Modern Physics Letters B}, 18\penalty0 (30):\penalty0
  1481--1505, 2004.

\bibitem[Kolokolnikov et~al.(2014)Kolokolnikov, Kevrekidis, and
  Carretero-Gonz{\'a}lez]{KoKeCa2014}
T.~Kolokolnikov, P.~G. Kevrekidis, and R.~Carretero-Gonz{\'a}lez.
\newblock A tale of two distributions: from few to many vortices in
  quasi-two-dimensional {B}ose--{E}instein condensates.
\newblock \emph{Proceedings of the Royal Society A: Mathematical, Physical and
  Engineering Sciences}, 470\penalty0 (2168):\penalty0 20140048, 2014.

\bibitem[Koukouloyannis et~al.(2014)Koukouloyannis, Voyatzis, and
  Kevrekidis]{KoVoKe2014}
V.~Koukouloyannis, G.~Voyatzis, and P.~G. Kevrekidis.
\newblock Dynamics of three noncorotating vortices in {B}ose-{E}instein
  condensates.
\newblock \emph{Physical Review E}, 89\penalty0 (4):\penalty0 042905--, 2014.

\bibitem[Kyriakopoulos et~al.(2014)Kyriakopoulos, Koukouloyannis, Skokos, and
  Kevrekidis]{KyKoSkKe2014}
N.~Kyriakopoulos, V.~Koukouloyannis, C.~Skokos, and P.~G. Kevrekidis.
\newblock Chaotic behavior of three interacting vortices in a confined
  {B}ose-{E}instein condensate.
\newblock \emph{Chaos: An Interdisciplinary Journal of Nonlinear Science},
  24\penalty0 (2):\penalty0 024410, 2014.

\bibitem[Marsden and Ratiu(1999)]{MaRa1999}
J.~E. Marsden and T.~S. Ratiu.
\newblock \emph{Introduction to Mechanics and Symmetry}.
\newblock Springer, 1999.

\bibitem[Marsden and Weinstein(1974)]{MaWe1974}
J.~E. Marsden and A.~Weinstein.
\newblock Reduction of symplectic manifolds with symmetry.
\newblock \emph{Reports on Mathematical Physics}, 5\penalty0 (1):\penalty0
  121--130, 1974.

\bibitem[Middelkamp et~al.(2010{\natexlab{a}})Middelkamp, Kevrekidis,
  Frantzeskakis, Carretero-Gonz\'alez, and Schmelcher]{MiKeFrCaSc2010b}
S.~Middelkamp, P.~G. Kevrekidis, D.~J. Frantzeskakis, R.~Carretero-Gonz\'alez,
  and P.~Schmelcher.
\newblock Bifurcations, stability, and dynamics of multiple matter-wave vortex
  states.
\newblock \emph{Phys. Rev. A}, 82:\penalty0 013646, 2010{\natexlab{a}}.

\bibitem[Middelkamp et~al.(2010{\natexlab{b}})Middelkamp, Kevrekidis,
  Frantzeskakis, Carretero-Gonz{\'a}lez, and Schmelcher]{MiKeFrCaSc2010a}
S.~Middelkamp, P.~Kevrekidis, D.~Frantzeskakis, R.~Carretero-Gonz{\'a}lez, and
  P.~Schmelcher.
\newblock Stability and dynamics of matter-wave vortices in the presence of
  collisional inhomogeneities and dissipative perturbations.
\newblock \emph{Journal of Physics B: Atomic, Molecular and Optical Physics},
  43\penalty0 (15):\penalty0 155303, 2010{\natexlab{b}}.

\bibitem[Middelkamp et~al.(2011)Middelkamp, Torres, Kevrekidis, Frantzeskakis,
  Carretero-Gonz\'alez, Schmelcher, Freilich, and Hall]{MiToKeFrCaScFrHa2011}
S.~Middelkamp, P.~J. Torres, P.~G. Kevrekidis, D.~J. Frantzeskakis,
  R.~Carretero-Gonz\'alez, P.~Schmelcher, D.~V. Freilich, and D.~S. Hall.
\newblock Guiding-center dynamics of vortex dipoles in {B}ose-{E}instein
  condensates.
\newblock \emph{Phys. Rev. A}, 84:\penalty0 011605, 2011.

\bibitem[Murray et~al.(2016)Murray, Groszek, Kuopanportti, and
  Simula]{MuGrKuSi2016}
A.~V. Murray, A.~J. Groszek, P.~Kuopanportti, and T.~Simula.
\newblock Hamiltonian dynamics of two same-sign point vortices.
\newblock \emph{Phys. Rev. A}, 93:\penalty0 033649, 2016.

\bibitem[Navarro et~al.(2013)Navarro, Carretero-Gonz{\'a}lez, Torres,
  Kevrekidis, Frantzeskakis, Ray, Altunta{\c s}, and
  Hall]{NaCaToKeFrRaAlHa2013}
R.~Navarro, R.~Carretero-Gonz{\'a}lez, P.~J. Torres, P.~G. Kevrekidis, D.~J.
  Frantzeskakis, M.~W. Ray, E.~Altunta{\c s}, and D.~S. Hall.
\newblock Dynamics of a few corotating vortices in {B}ose-{E}instein
  condensates.
\newblock \emph{Physical Review Letters}, 110\penalty0 (22):\penalty0 225301--,
  2013.

\bibitem[Neely et~al.(2010)Neely, Samson, Bradley, Davis, and
  Anderson]{NeSaBrDaAn2010}
T.~W. Neely, E.~C. Samson, A.~S. Bradley, M.~J. Davis, and B.~P. Anderson.
\newblock Observation of vortex dipoles in an oblate {B}ose-{E}instein
  condensate.
\newblock \emph{Phys. Rev. Lett.}, 104:\penalty0 160401, 2010.

\bibitem[Neu(1990)]{Ne1990}
J.~C. Neu.
\newblock Vortices in complex scalar fields.
\newblock \emph{Physica D: Nonlinear Phenomena}, 43\penalty0 (2):\penalty0
  385--406, 1990.

\bibitem[Ohsawa(2019)]{Oh2019d}
T.~Ohsawa.
\newblock Symplectic reduction and the {L}ie--{P}oisson shape dynamics of {$N$}
  point vortices on the plane.
\newblock \emph{Nonlinearity}, 32\penalty0 (10):\penalty0 3820--3842, 2019.

\bibitem[Ohsawa()]{Point_Vortex_Stability-Plane}
T.~Ohsawa.
\newblock Nonlinear stability of relative equilibria in planar {$N$}-vortex
  problem.
\newblock \emph{arXiv:2406.12144}.

\bibitem[Pelinovsky and Kevrekidis(2011)]{PeKe2011}
D.~Pelinovsky and P.~Kevrekidis.
\newblock Variational approximations of trapped vortices in the large-density
  limit.
\newblock \emph{Nonlinearity}, 24\penalty0 (4):\penalty0 1271, 2011.

\bibitem[Pismen and Rubinstein(1991)]{PiRu1991}
L.~M. Pismen and J.~Rubinstein.
\newblock Motion of vortex lines in the {G}inzburg--{L}andau model.
\newblock \emph{Physica D: Nonlinear Phenomena}, 47\penalty0 (3):\penalty0
  353--360, 1991.

\bibitem[Seman et~al.(2010)Seman, Henn, Haque, Shiozaki, Ramos, Caracanhas,
  Castilho, Castelo~Branco, Tavares, Poveda-Cuevas, Roati, Magalh\~aes, and
  Bagnato]{SeHeHaShRaCaCaCaTaPoRoMaBa2010}
J.~A. Seman, E.~A.~L. Henn, M.~Haque, R.~F. Shiozaki, E.~R.~F. Ramos,
  M.~Caracanhas, P.~Castilho, C.~Castelo~Branco, P.~E.~S. Tavares, F.~J.
  Poveda-Cuevas, G.~Roati, K.~M.~F. Magalh\~aes, and V.~S. Bagnato.
\newblock Three-vortex configurations in trapped {B}ose-{E}instein condensates.
\newblock \emph{Phys. Rev. A}, 82:\penalty0 033616, 2010.

\bibitem[Torres et~al.(2011)Torres, Kevrekidis, Frantzeskakis,
  Carretero-Gonz{\'a}lez, Schmelcher, and Hall]{ToKeFrCaScHa2011}
P.~Torres, P.~Kevrekidis, D.~Frantzeskakis, R.~Carretero-Gonz{\'a}lez,
  P.~Schmelcher, and D.~Hall.
\newblock Dynamics of vortex dipoles in confined {B}ose--{E}instein
  condensates.
\newblock \emph{Physics Letters A}, 375\penalty0 (33):\penalty0 3044--3050,
  2011.
\newblock ISSN 0375-9601.

\end{thebibliography}
\bibliographystyle{plainnat}

\end{document}